\documentclass{article}
\usepackage[utf8]{inputenc}
\usepackage[T1]{fontenc}
\usepackage[english]{babel}
\usepackage{amsmath,amsthm,amssymb,amsfonts}
\usepackage{xspace}
\usepackage{hyperref}
\usepackage{cleveref}
\usepackage{verbatim}
\usepackage{thm-restate}
\usepackage{tikz}
\usepackage[hmargin=2.5cm,vmargin=3.8cm]{geometry}
\usepackage[textsize=footnotesize,color=green!40]{todonotes}

\usepackage{bm}
\newcommand{\nd}{$\mathsf{nd}$}
\newcommand{\dts}{$\mathsf{dts}$}
\newcommand{\cvd}{$\mathsf{cvd}$}

\newcommand{\FPT}{$\mathsf{FPT}$ }

\newcommand{\CNR}{\textsc{Cops and Robber}\xspace}
\newcommand{\RBDS}{\textsc{Red-Blue Dominating Set}\xspace}
\newcommand{\CR}{\textsc{CnR}\xspace}
\newcommand{\CFR}{\textsc{Cops and Fast Robber}\xspace}
\newcommand{\LCR}{\textsc{Lazy CnR}\xspace}
\newcommand{\ACR}{\textsc{Fully Active CnR}\xspace}
\newcommand{\CAR}{\textsc{Cops and Attacking Robber}\xspace}
\newcommand{\SCR}{\textsc{Surrounding CnR}\xspace}

\newcommand{\NPoly}{\textsf{NP} $\subseteq$ \textsf{coNP}$/$\textsf{poly}\xspace}

\newenvironment{proofsketch}{%
  \proof}{\endproof}

\newenvironment{proofofclaim}{%
  \proof}{\endproof}

\newcommand{\R}{\small \mathcal{R}}
\newcommand{\C}{\small \mathcal{C}}
\newtheorem{theorem}{Theorem}
\crefname{theorem}{theorem}{theorems}
\newtheorem{lemma}[theorem]{Lemma}
\crefname{lemma}{lemma}{lemmas}
\newtheorem{proposition}[theorem]{Proposition}

\crefname{proposition}{proposition}{propositions}
\crefname{result}{result}{results}
\newtheorem{corollary}[theorem]{Corollary}
\crefname{corollary}{corollary}{corollaries}

\crefname{fact}{fact}{facts}
\newtheorem{observation}[theorem]{Observation}
\crefname{observation}{observation}{observations}
\newtheorem{question}[theorem]{Question}
\crefname{question}{question}{questions}
\newtheorem{claim}[theorem]{Claim}
\crefname{claim}{claim}{claims}
\newtheorem{note}[theorem]{Note}
\crefname{note}{note}{notes}

\crefname{conj}{conjecture}{conjectures}

\crefname{definition}{definition}{definitions}

\crefname{remark}{remark}{remarks}

\newtheorem{theorem1}{Theorem}
\newtheorem{RR}[theorem1]{Reduction Rule}

\tikzstyle{noeud}=[circle,inner sep=2, minimum size =3 pt, line width = 1pt, draw=black, fill=white]

\newcommand{\Pb}[4]{%
\begin{center}
  \begin{tabular}{|l|}%
  \hline
    \begin{minipage}[c]{0.95\textwidth}
      \smallskip%
      \par\noindent%
      #1%
      \par\noindent%
      \textbf{\textsf{Input}}: #2%
      \par\noindent%
      \textbf{\textsf{#3}}: #4 
      \smallskip%
      \par\noindent%
    \end{minipage}
  \\\hline
  \end{tabular}%
\end{center}
}%

\title{Parameterized Analysis of the Cops and Robber Problem}
\author{
{Harmender Gahlawat} and {Meirav Zehavi}\\
\mbox{}\\
{\small  Ben-Gurion University of the Negev, Beersheba, Israel}\\}

\begin{document}

\maketitle

\begin{abstract}
\textit{Pursuit-evasion games} have been intensively studied for several decades due to their numerous applications in artificial intelligence, robot motion planning, database theory, distributed computing, and algorithmic theory. \textsc{Cops and Robber} (\CR) is one of the most well-known pursuit-evasion games played on graphs, where multiple \textit{cops} pursue a single  \textit{robber}. The aim is to compute the \textit{cop number} of a graph, $k$, which is the minimum number of cops that ensures the \textit{capture} of the robber. 

From the viewpoint of parameterized complexity, \CR is W[2]-hard parameterized by $k$~[Fomin et al., TCS, 2010]. 
Thus, we study structural parameters of the input graph. We begin with the \textit{vertex cover number} ($\mathsf{vcn}$). First, we establish that $k \leq \frac{\mathsf{vcn}}{3}+1$. Second, we prove that \CR parameterized by $\mathsf{vcn}$ is \FPT by designing an exponential kernel. We complement this result by showing that it is unlikely for \CR parameterized by $\mathsf{vcn}$ to admit a polynomial compression. We extend our exponential kernels to the parameters \textit{cluster vertex deletion number} and \textit{deletion to stars number}, and design a linear vertex kernel for \textit{neighborhood diversity}.  Additionally, we extend all of our results to several well-studied variations of \CR. 
\end{abstract}

\section{Introduction}\label{S:Intro}
In \textit{pursuit-evasion}, a set of agents, called \textit{pursuers}, plan to catch one or multiple \textit{evaders}. Classically, pursuit-evasion games were played on geometric setups, where pursuers and evaders move on the plane following some rules~\cite{diffrentialPursuitBook,diffPursuitBook,bookPursuit-EvasionGames}. Parsons~\cite{parsons1,parsons2} formulated pursuit-evasion on graphs to model the search for a person trapped in caves, giving rise to the field of graph searching. Since then, pursuit-evasion has been studied extensively, having applications in artificial intelligence~\cite{app1}, robot motion planning~\cite{appRobotics,appRobotics2}, constraint satisfaction and database theory~\cite{appDB1,appDB2,appDB3}, distributed computing~\cite{appDC2,appDC} and network decontamination~\cite{nisseRecent}, and significant implications in graph theory and algorithms~\cite{seymour,appVertexSeparation,cop-decs,appLogic}. 
 

\CNR (\CR) is one of the most intensively studied pursuit-evasion games on graphs, where a set of cops pursue a single robber. Players move in discrete time steps alternately, starting with the cops. In each move, a player can move to an adjacent vertex, and the cops win by \textit{capturing} the robber (i.e., if a cop and the robber occupy the same vertex). The goal is to compute the \textit{cop number} of a graph $G$, denoted $\mathsf{c}(G)$, which is the minimum number of cops required to win in $G$. We define the game formally in Section~2. \CR is well studied in the artificial intelligence literature under the name \textsc{Moving Target Pursuit} (\textsf{MTP})~\cite{MTS}, where we consider sub-optimal but faster strategies from an applicative point of view. The results have found numerous applications in game design, police chasing, path planning, and robot motion planning~\cite{MTP5,MTP3,MTP2,MTP4,MTP1}. 

Determining the parameterized complexity of games is a well-studied research topic~\cite{paraGamesScott,paraConnectionGames,paraPositionalGames}. 
Most pursuit-evasion games are, in fact, AW[*]-hard~\cite{paraPursuit-Evasion}. In particular, \CR is W[2]-hard parameterized by $\mathsf{c}(G)$~\cite{fomin}. Thus, we consider structural parameterizations, focusing on \textit{kernelization}, also known as polynomial-time preprocessing with a parametric guarantee. Due to the profound impact of preprocessing, kernelization was termed ``the lost continent of polynomial time''~\cite{kernelApplication}. We begin with the most studied structural parameter in parameterized complexity: the \textit{vertex cover number} ($\mathsf{vcn}$) of the input graph. We bound $\mathsf{c}(G)$ in terms of $\mathsf{vcn}$, as well as achieve both positive and negative results concerning the kernelization complexity of \CR parameterized by $\mathsf{vcn}$. We generalize our kernelization results to the smaller parameters \textit{cluster vertex deletion number} (\cvd) and \textit{deletion to stars number} (\dts), as well as to the parameter \textit{neighborhood diversity} (\nd). Furthermore, we extend all our results to several well-studied variants of \CR.


The choice of $\mathsf{vcn}$ as a parameter to study pursuit-evasion games is natural due to various scenarios where $\mathsf{vcn}$ is significantly smaller than the graph size. For example, this includes scenarios where we model the existence of one or few (possibly interconnected) central hubs---for illustration, suppose an intruder is hiding in a system of buildings where we have only few corridors but a large number of rooms, or suppose we have few virtual servers with many stations (e.g., of private users) that can communicate only with the servers. Furthermore, $\mathsf{vcn}$ is one of the most efficiently computable parameters from both approximation~\cite{bookApprox} and parameterized~\cite{bookParameterized} points of view, making it fit from an applicative perspective even when a vertex cover is not given along with the input. Moreover, $\mathsf{vcn}$ is the best choice for proving negative results---indeed, our negative result on the kernelization complexity of $\CR$~for $\mathsf{vcn}$ implies the same for many other well-known smaller parameters such as treewidth, treedepth and feedback vertex set~\cite{bookKernelization}. One shortcoming of $\mathsf{vcn}$ as a parameter is that it is very high for some simple (and easy to resolve) dense graphs like complete graphs. However, we generalize our kernel to \cvd, which is small for these dense graphs, and for \dts. Furthermore, we design a \textit{linear} kernel for the well-studied parameter \nd. We further discuss the utility of our kernels in the Conclusion.


\subsection{Brief Survey}

\CR was independently introduced by Quilliot~\cite{qui} and by  Nowakowski and Winkler~\cite{nowakowski} with exactly one cop\footnote{In fact, a specific instance of \CR on a specific graph was given as a puzzle in Problem~395 of the book Amusements in Mathematics~\cite{bookAmusements} already in 1917.}. Aigner and Fromme~\cite{aigner} generalized the game to multiple cops and defined the \textit{cop number} of a graph. 
The notion of cop number and some fundamental techniques introduced by Aigner and Fromme~\cite{aigner} yielded a plethora of results on this topic. For more details, we refer the reader to the book~\cite{bonatobook}.

The computational complexity of finding the cop number of a graph has been a challenging subject of research. On the positive side, Berarducci and Intrigila~\cite{berarducci} gave a backtracking algorithm that decides whether $G$ is $k$-copwin in $\mathcal{O}(n^{2k+1})$ time.  
On the negative side, Fomin et al.~\cite{fomin} proved that determining whether $G$ is $k$-copwin is NP-hard, and  W[2]-hard parameterized by $k$. Moreover, Mamino~\cite{mamino} showed that the game is PSPACE-hard, and later, Kinnersley~\cite{kinnersley} proved that determining the cop number of a graph is, in fact, EXPTIME-complete. Recently, Brandt et al.~\cite{capture-bound} provided fine-grained lower bounds, proving that the time complexity of any algorithm for \CR is $\Omega(n^{k-o(1)})$ conditioned on \textit{Strong Exponential Time Hypothesis} ($\mathsf{SETH}$), and $2^{\Omega (\sqrt{n})}$ conditioned on \textit{Exponential Time Hypothesis} ($\mathsf{ETH}$).

Since \CR admits an XP-time algorithm, it is sensible to bound the cop number for various graph classes or by some structural parameters. Nowadays, we know that the cop number is 3 for the class of planar graphs~\cite{aigner} and toroidal graphs~\cite{lehner}, 9 for unit-disk graphs~\cite{udg}, 13 for string graphs~\cite{ourString}, and is bounded for bounded genus graphs~\cite{bowler} and  minor-free graphs~\cite{andreae}. Moreover, it is known that the cop number of a graph $G$ is at most $\frac{\mathsf{tw}(G)}{2}+1$~\cite{joret}, where $\mathsf{tw}(G)$ denotes the treewidth of $G$, and at most $\mathsf{cw}(G)$~\cite{fomin}, where $\mathsf{cw}(G)$ denotes the clique-width of $G$.



\subsection{Our Contribution}
We conduct a comprehensive analysis of \CR parameterized by $\mathsf{vcn}$. We start by bounding the cop number of a graph by its vertex cover number:
\begin{restatable}{theorem}{VCBound}\label{th:VCbound}
For a graph $G$, $\mathsf{c}(G) \leq \frac{\mathsf{vcn}}{3}+1$.
\end{restatable}

The proof is based on the application of three reduction rules. Each of our rules controls its own cop that, in particular, guards at least three vertices that belong to the vertex cover. Once our rules are no longer applicable, we exhibit that the remaining unguarded part of the graph is of a special form. In particular, we exploit this special form to prove that, now, the usage of only two additional cops suffices.
We complement Theorem~\ref{th:VCbound} with an argument (Lemma~\ref{L:Best}) that it might be difficult to improve this bound further using techniques similar to ours.

Second, we prove that \CR parameterized by $\mathsf{vcn}$ is $\mathsf{FPT}$ by designing a kernelization algorithm:

\begin{restatable}{theorem}{VCKernel}\label{th:Kernel}
\CR parameterized by $\mathsf{vcn}$ admits a kernel with at most $\mathsf{vcn}+ \frac{2^\mathsf{vcn}}{\sqrt{\mathsf{vcn}}}$ vertices.
\end{restatable}

Our kernel is also based on the application of reduction rules. However, these rules are very different than those used for the proof of Theorem 1. While our main rule is quite standard in kernelization (involving the removal of, in particular, false twins), the proof of its correctness is (arguably) not.
Theorem~\ref{th:Kernel}, along with Theorem~\ref{th:VCbound} and an XP-algorithm (Proposition~\ref{P:XP}), gives the following immediate corollary:

\begin{restatable}{corollary}{VCFPT}\label{C:paraAlgo}
\CR is \FPT parameterized by $\mathsf{vcn}$, and is solvable in $ \big{(}\mathsf{vcn}+\frac{2^\mathsf{vcn}}{\sqrt{\mathsf{vcn}}}\big{)}^{\!\frac{\mathsf{vcn}}{3}+2}\!\cdot n^{\mathcal{O}(1)}$ time.
\end{restatable}

We complement our kernel by showing that it is unlikely for \CR to admit polynomial compression, by providing a \textit{polynomial parameter transformation} from \RBDS.   In particular, our reduction makes non-trivial use of a known construction of a special graph having high girth and high minimum degree.

\begin{restatable}{theorem}{VCNPoly}\label{th:PolyCompression}
\CR parameterized by $\mathsf{vcn}$ does not admit polynomial compression, unless \NPoly.
\end{restatable}

Next, we present a linear kernel for \CR parameterized by neighbourhood diversity:
\begin{restatable}{theorem}{ND}\label{T:nd}
   \CR parameterized by $\mathsf{nd}$ admits a kernel with at most $\mathsf{nd}$ vertices. 
\end{restatable}



On the positive side, we extend our exponential kernel to two smaller structural parameters, \cvd~and  \dts:
\begin{restatable}{theorem}{VClique}\label{T:vertexClique}
\CR parameterized by \cvd~admits a kernel with at most  $2^{2^\mathsf{cvd} + \sqrt{\mathsf{cvd}}}$ vertices. Moroever, \CR parameterized by \dts~admits a kernel with at most  $2^{2^\mathsf{dts} + \mathsf{dts}^{1.5}}$ vertices.
\end{restatable}


Several variants of \CR have been studied due to their copious applications. We extend our results, parameterized by \textsf{vcn}, to some of the most well-studied ones. We define these variants (and used notations) in Section~\ref{S:preliminaries}. We first bound the cop number of these variants by  $\mathsf{vcn}$:

\begin{theorem} \label{T:variationBound}
For a graph $G$: (1) $\mathsf{c}_{lazy} \leq \frac{\mathsf{vcn}}{2} +1$; (2) $\mathsf{c}_{attack} \leq \frac{\mathsf{vcn}}{2} +1$; (3) $\mathsf{c}_{active}(G) \leq \mathsf{vcn}$; (4) $\mathsf{c}_{surround}(G) \leq \mathsf{vcn}$; (5) $\mathsf{c}_s(G) \leq \mathsf{vcn}$ (for any value of $s$); (6) for a strongly connected orientation $\overrightarrow{G}$ of $G$, $\mathsf{c}(\overrightarrow{G}) \leq \mathsf{vcn}$.
\end{theorem}


We also extend our exponential kernel to these variants:
\begin{restatable}{theorem}{LAtKernel}\label{T:LAtKernel}
\CAR and \LCR parameterized by $\mathsf{vcn}$ admit a kernel with at most $\mathsf{vcn}+\frac{2^\mathsf{vcn}}{\sqrt{\mathsf{vcn}}}$ vertices. Moreover, \CR on strongly connected directed graphs admits a kernel with at most $3^\mathsf{vcn}+\mathsf{vcn}$ vertices.
\end{restatable}

Then, we present a slightly more general kernelization that works for most variants of the game. In particular, we define a new variant of the game (in Section~\ref{S:preliminaries}), \textsc{Generalized CnR},that generalizes various well studied variants of \CR. We have the following result that proves that \textsc{Generalized CnR} parameterized by $\mathsf{vcn}$ admits an exponential kernel.

\begin{restatable}{theorem}{GenKernel}\label{T:general}
\textsc{Generalized CnR} parameterized by $\mathsf{vcn}$ admits a kernel with at most $\mathsf{vcn}+\mathsf{vcn}\cdot 2^{\mathsf{vcn}}$ vertices.
\end{restatable}

Then, we show that the same kernelization algorithm also provides us the following result:
\begin{theorem}\label{Th:active}
\ACR, \CFR, and \SCR parameterized by $\mathsf{vcn}$ admit a kernel with at most $\mathsf{vcn}+\mathsf{vcn}\cdot 2^\mathsf{vcn}$ vertices.
\end{theorem}

Finally, we complement our exponential kernels for these variants by arguing about their incompressibility:

\begin{theorem}\label{th:variantPolyCOmpression}
\LCR, \CAR, \CFR, \ACR, and \CR on strongly connected directed and oriented graphs parameterized by $\mathsf{vcn}$ do not admit a polynomial compression, unless \NPoly.
\end{theorem}


\subsection{Additional Related Works}
For a graph with girth at least 5, the cop number is lower bounded by the minimum degree of the graph~\cite{aigner}. As implied by the lower bound for the \textit{Zarankiewicz problem}~\cite{bollobas}, an extremal graph with girth 5 has $\Omega(n^{3/2})$ edges.  
In a graph with $\Omega(n^{3/2})$ edges, if there is a vertex whose degree is smaller than $c\sqrt{n}$, for an appropriate constant $c$, then we can remove it and still get a smaller graph with $\Omega(n^{3/2})$ edges. 
Hence, eventually, every vertex has degree $\Omega(\sqrt{n})$. 
Therefore, the cop number of such a graph is $\Omega(\sqrt{n})$. 
Meyniel~\cite{mey} conjectured this to be tight, that is, $\mathcal{O}(\sqrt{n})$ cops are sufficient to capture the robber in any connected graph. 
This is probably the deepest conjecture in this field (see~\cite{baird}). Since then, several attempts have been made to bound the cop number of general graphs~\cite{chini,mey-peng,mey-scott}. Although these results establish that the $\mathsf{c}(G) = o(n)$, even the question  whether $c(G) = \mathcal{O}(n^{1-\epsilon})$, for $\epsilon >0$, remains open. 

Many graph classes have unbounded cop number. The graph classes for which the cop number is $\Omega(\sqrt{n})$ are called \textit{Meyniel extremal}. These include bipartite graphs~\cite{baird}, subcubic graphs~\cite{mey-boundedDegree}, and polarity graphs~\cite{designs}. Meyniel's conjecture was also considered for random graphs~\cite{mey-random2,mey-random1,mey-random3}.

Lastly, we remark that variations of \CR vary mainly depending on the capabilities of the cops and the robber. Some of these variations were shown to have correspondence with several width measures of graphs like treewidth~\cite{seymour}, pathwidth~\cite{parsons1}, tree-depth~\cite{depth}, hypertree-width~\cite{adler}, cycle-rank~\cite{depth}, and directed tree-width~\cite{dtwidth}. Moreover, Abraham et al.~\cite{cop-decs} defined the concept of a  \textit{cop-decomposition}, which is based on the cop strategy in the \CR game on minor-free graphs provided by Andreae~\cite{andreae}, and showed that it has significant algorithmic applications.

\section{Preliminaries}\label{S:preliminaries}

For $\ell \in \mathbb{N}$, let $[\ell]$ = $\{1,\ldots, \ell \}$. Whenever we mention $\frac{a}{b}$, we mean $\lceil \frac{a}{b} \rceil$. 

 \subsection{Graph Theory}
For a graph $G$, we denote its vertex set by $V(G)$ and edge set by $E(G)$. We denote the size of $V(G)$ by $n$ and size of $E(G)$ by $m$. In this paper, we consider finite, connected\footnote{The cop number of a disconnected graph is the sum of the cop numbers of its components; hence, we assume connectedness.}, and simple graphs.
Let $v$ be a vertex of a graph $G$. Then, by $N(v)$ we denote the \textit{open neighbourhood} of $v$, that is, $N(v)= \{u ~|~ uv \in E(G)\}$. 
By $N[v]$ we denote the \textit{close neighbourhood} of $v$, that is, $N[v] = N(v) \cup \{v\}$. For $X \subseteq V(G)$, we define $N_X(v) = N(v) \cap X$ and $N_X[v] = N[v] \cap X$. We say that $v$ \textit{dominates} $u$ if $u\in N[v]$. The \textit{girth} of a graph $G$ is the length of a shortest cycle contained in $G$.  
A {\em $u,v$-path} is a path with endpoints $u$ and $v$. A path is \textit{isometric} if it is a shortest path between its endpoints. For $u,v\in V(G)$, let $d(u,v)$ denote the length of a shortest $u,v$-path.

Let $G$ be a graph and $U\subseteq V(G)$.  Then, $G[U]$ denotes the subgraph of $G$ induced by $U$. A set $U \subseteq V(G)$ is a \textit{vertex cover} if $G[V(G) \setminus U]$ is an independent set. The minimum cardinality of a vertex cover of $G$ is its \textit{vertex cover number} ($\mathsf{vcn}$).  Moreover,  $U$ is a \textit{cluster vertex deletion set} if $G[V(G) \setminus U]$ is a disjoint union of cliques. The minimum size of a cluster vertex deletion set of a graph is its \textit{cluster vertex deletion number} ($\mathsf{cvd}$). Additionally, $U$ is a \textit{deletion to stars set} if $G[V(G) \setminus U]$ is a disjoint union of star graphs. The minimum size of a deletion to stars set of a graph is its \textit{deletion to stars number} ($\mathsf{dts}$). Two vertices $u,v \in V(G)$ have the \textit{same type} if and only if $N(v)\setminus \{u\} = N(u) \setminus \{v\}$. A graph $G$ has \textit{neighborhood diversity} at most $w$ if there exists a partition of $V(G)$ into at most $w$ sets, such that all the vertices in each set have the same type.

\subsection{\CNR}
\CR is a two-player perfect information pursuit-evasion game played on a graph. 
One player is referred as \textit{cop player} and controls a set of \textit{cops}, and the other player is referred as \textit{robber player} and controls a single \textit{robber}. 
The game starts with the cop player placing each cop on some vertex of the graph, and multiple cops may simultaneously occupy the same vertex. Then, the robber player places the robber on a vertex. 
Afterwards, the cop player and the robber player make alternate moves, starting with the cop player. 
In the cop player move, the cop player, for each cop, either moves it to an adjacent vertex (along an edge) or keeps it on the same vertex. In the robber player move, the robber player does the same for the robber. For simplicity, we will say that the cops (resp., robber) move in a cop (resp., robber) move instead of saying that the cop (resp., robber) player moves the cops (resp., robber). Throughout, we denote the robber by $\R$.

A situation where one of the cops, say, $\C$, occupies the same vertex as $\R$ is a \textit{capture}. (We also say that the $\C$ captures $\R$ and that $\R$ is captured by $\C$.) The cops win if they have a strategy to capture $\R$, and $\R$ wins if it has a strategy to evade a capture indefinitely. A graph $G$ is \textit{$k$-copwin} if $k$ cops have a winning strategy in $G$.
The \textit{cop number} of $G$, denoted $\mathsf{c}(G)$, is the minimum $k$ such that $G$ is $k$-copwin. For brevity, $G$ is said to be \textit{copwin} if it is $1$-copwin (i.e. $\textsf{c}(G) = 1$). 
Accordingly, we have the following decision version of the problem.




\Pb{\CNR}{ A graph $G$, and an integer $k \in \mathbb{N}$}{Question}{Is $G$ $k$-copwin?}

We say that some cops \textit{guard} a subgraph $H$ of $G$ if $\R$ cannot enter $H$ without getting captured by one of these cops in the next cop move. We shall use the following result:
\begin{proposition}[\cite{aigner}] \label{P:aigner}
Let $P$ be an isometric path in $G$. Then one cop can guard $P$ after a finite number of rounds/cop moves.
\end{proposition}

Currently, the best known algorithm to decide whether $G$ is $k$-copwin is by Petr et al.~\cite{AlgoCR}:

\begin{proposition}[\cite{AlgoCR}]\label{P:XP}
\CR is solvable in $\mathcal{O}(kn^{k+2})$ time.
\end{proposition}

If a cop $\C$ occupies a vertex $v$, then $\C$ \textit{attacks} $N[v]$. A vertex $u$ is \textit{safe} if it is not being attacked by any cop. If $\R$ is on a vertex that is not safe, then $\R$ is \textit{under attack}.


\subsection{Variations of \CR}
Several variations of \CR have been studied in the literature, differing mainly in the rules of movements of agents, the definition of the capture, and the capabilities of the agents. We provide below the definitions of the games considered in this paper. We list below some of the primary properties of the gameplay in which these variations differ:

\begin{enumerate}
    \item \textit{Speed of agents}: If an agent has speed $s$, where $s\in \mathbb{N}$, then the agent can move along at most $s$ edges in its turn. We note that a robber with speed $s$ cannot move over a cop, that is, the robber can move along a path of length at most $s$ not containing any cop, in its turn.
    
    \item \textit{Lazy/active/flexible cops}:
    Let $C$ be the set of cops and let $A\cup F \cup L$ be a partition of the set of cops such that $A$ is the set of \textit{active} cops, $F$ be the set of \textit{flexible} cops, and $L$ be the set of \textit{lazy} cops. Then, in each cop move, at most one cop from $L$ can make a move, each cop from $A$ must make a move, and each cop from $F$ can either make a move or stay on the same vertex. Unless mentioned otherwise, all cops are assumed to be flexible. 

    \item \textit{Reach of cops:}
    If a cop $\C_i$ has \textit{reach} $\lambda_i$, then $\R$ cannot access a vertex that is at a distance at most $\lambda_i$ from the vertex occupied by $\C_i$. Here, think of the cop $\C_i$ as having a gun with range $\lambda_i$. Hence, if $\C_i$ can reach a vertex that is at most distance $\lambda_i$ from the robber's vertex at the end of a cop move, then $\C_i$ can shoot $\R$, and the cops win. Similarly, on a robber move, even if $\R$ has speed $s$, then it can move only along a path of length at most $s$ that does not contain any vertex that is at a distance at most $\lambda_i$ from $\C_i$. In \CR, for each cop $\C_i$, $\lambda_i = 0$.   
    
    \item \textit{Visible/invisible robber}: If the robber is \textit{visible}, then the cops know the position of the robber. If the robber is \textit{invisible}, then the cops do not know the position of the robber. Moreover, we say that cops have \textit{$d$-visibility} if cops can see the position of the robber only if it is at most $d$ edges away from at least one of the cops.

\end{enumerate}

Next, we define the variants of \CR for which we will extend our results.

\medskip
\noindent\textbf{\LCR:} \LCR~\cite{offner} is one the the most well-studied variants of \CR games~\cite{bal,sim}. In this variant, the cops are lazy, that is, at most one cop can move during a cops' turn. This restricts the ability of the cops with respect to the classical version. The minimum number of lazy cops that can ensure a capture in a graph $G$ is known as the \textit{lazy cop number} and is denoted by $\mathsf{c}_{lazy}(G)$. Clearly, $\mathsf{c}(G) \leq \mathsf{c}_{lazy}(G)$, as $\mathsf{c}_{lazy}(G)$ cops can capture the robber in the classical version (using the winning strategy of the \LCR game). We remark that this game is also studied with the name \textit{one-cop-moves} game~\cite{lazyplanar,lazywang}.



\medskip
\noindent\textbf{\CAR:}
In \CAR~\cite{bonatocar}, the robber is able to \textit{strike back} against the cops. If on a robber's turn, there is a cop in its neighborhood, then the robber can attack the cop and \textit{eliminate} it from the game. However, if more than one cop occupy a vertex and the robber attacks them, then only one of the cops gets eliminated, and the robber gets captured by one of the other cops on that vertex. The cop number for capturing an attacking robber on a graph $G$ is denoted by $\mathsf{c}_{attack}(G)$, and is referred to as the \textit{attacking cop number} of $G$. Clearly, $\mathsf{c}(G) \leq \mathsf{c}_{attack}(G) \leq 2 \cdot \mathsf{c}(G)$, as, on the one hand, $\mathsf{c}_{attack}(G)$ cops can capture the robber in the classical version. On the other hand, if we play the attacking version with $2\cdot \mathsf{c}(G)$ cops using the strategy of the classical variant with the only difference that there are always at least two cops on a vertex, then the cops have a winning strategy. 

\medskip
\noindent\textbf{\ACR:}
In the game of \ACR~\cite{active}, each cop as well as the robber are active, that is,  in a cop/robber move, each cop/robber has to move to an adjacent vertex. The \textit{active cop number} of a graph $G$, denoted by $\mathsf{c}_{active}(G)$, is the minimum number of cops that can ensure capture in this game. It is easy to see that $\mathsf{c}_{active}(G) \leq 2\cdot \mathsf{c}(G)$, as if we keep one extra cop adjacent to each cop in the winning strategy for \CR, then whenever some cop has to skip a move, it can simply do so by switching with the extra cop adjacent to it. 

\medskip
\noindent\textbf{\SCR:}
In the game of \SCR~\cite{surrounding1,surrounding2}, the definition of capture is different. In this game, a cop and the robber can occupy the same vertex of the graph during the game, but the robber cannot end its turn by remaining at a vertex occupied by some cop. The cops win by \textit{surrounding} the robber, that is, if the robber occupies a vertex $v$, then there is a cop at each vertex $u\in N(v)$. The \textit{surrounding cop number} for a graph $G$ is denoted as $\mathsf{c}_{surround}(G)$. It is easy to see that $\mathsf{c}_{surround}(G) \geq \delta(G)$, where $\delta(G)$ is the \textit{minimum degree} of the graph.

\medskip
\noindent\textbf{\CFR:}
In the game of \CFR~\cite{fomin}, the robber can move faster than the cops. If $\R$ has speed $s$, then it can move along a path with at most $s$ edges not containing any cop. The minimum number of cops that can ensure capture of 
a fast robber with speed $s$ in a graph $G$ is denoted by $\mathsf{c}_s(G)$. For $s \geq 2$, deciding whether $\mathsf{c}_s(G)\leq k$ is NP-hard as well as W[2]-hard even when input graph $G$ is restricted to be a split graph~\cite{fomin}. The game of \CFR is well-studied~\cite{fast2,balister,fast3}.

\medskip
\noindent\textbf{\CR on Directed Graphs:}
The game of \CR is also well-studied for oriented/directed graphs~\cite{loh,hosseini,mohar}. The game is played on a directed graph $\overrightarrow{G}$, and the players can only move along the orientation of the arcs. 

Finally, we define a variant of \CR that generalizes many well-studied variants of \CR:

\medskip
\noindent\textbf{\textsc{Generalized} \CR:}
Consider the following generalized version of \CR. Here the input is $(G,\C_1,\ldots,\C_k, \R)$ where  each cop $\C_i$ has \textit{speed} $s_i$ (possibly different for each cop) and $\R$ has speed $s_R$. Moreover, each cop can be either forced to be \textit{active} (all active cops have to move in each turn), \textit{lazy} (at most one lazy cop moves in each turn), or \textit{flexible} (a flexible cop can either move or stay on the same vertex in its move). Moreover, the robber can also be forced to be either lazy or flexible. Furthermore, each cop $\C_i$ can have \textit{reach} $\lambda_i$ (possibly different for each cop). This game generalizes several well-studied variants of \CR along with \CFR, \ACR, and \textsc{Cops and Robber From a Distance}~\cite{bonatoDistance}. It also generalizes the game of~\cite{MTP3,app1}. 

Finally, we note that we assume the notion of ``being active'' to be defined only when the agent has speed $s=1$.  But, this notion can be defined in multiple ways if the agent has speed $s>1$: the player might have to move at least $s'\leq s$ edges, the player may have to move to a vertex at a distance at least $s' \leq s$ from the current vertex, the player may or may not be allowed to repeat edges, and so on. We remark that our kernelization result for \textsc{Generalized CnR} can be made to work, with some changes, considering any of these notions discussed.  

\subsection{An XP Algorithm for Variants}
For graph searching games, there is a standard technique to get an XP-time algorithm with running time $n^{\mathcal{O}(k)}$ (where $n$ is the size of input graph and the question is whether $k$ cops have a winning strategy). This technique involves generating a \textit{game graph} where each vertex represents a possible placement of all the agents on the vertices of $G$. Since $k$ cops and a single robber can have $n^{k+1}$ possible placements on $G$, the game graph has $n^{k+1}$ vertices. The following step is to mark all of the \textit{winning states} (that is, where the robber is captured). Afterwards, we use an algorithm to keep adding states to the set of winning states in the following manner. On a cop move, from a given state $S$, if there exists a movement of cops that can change the game state $S$ to a winning state, we add $S$ to the winning states. On a robber move, for a game state $S$, if all the possible moves of the robber lead to a winning state, we add $S$ to the winning state.  Finally, if there exists a position of $k$ cops such that, for any position of the robber, these states are in winning states, we declare that $k$ cops have a winning strategy in $G$. It is easy to see that this algorithm can be implemented in $n^{\mathcal{O}(k)}$ time. 

Petr, Portier, and Versteegan~\cite{AlgoCR} gave an implementation of this algorithm, for \CR, that runs in $\mathcal{O}(kn^{k+2})$ time. It is not difficult to see that this algorithm can be made to work for all the variants we discussed by changing the rules to navigate between game states. For \CAR, the only extra consideration is that if $\R$ attacks a cop (among $k$ cops) and does not get captured in the next cop move, then we have a game state, say, $S'$, with $k-1$ cops and one robber, where the placement of these agents is a subset of a placement of $k+1$ agents in one of the original game states, and hence we prune $S'$. Thus, we have the following proposition.

\begin{proposition}\label{P:generalXP}
For any variant of \CR considered in this paper, an instance $(G,k)$ can be solved in $\mathcal{O}(kn^{k+2})$ time.
\end{proposition}

\subsection{Parameterized complexity}
In the framework of parameterized complexity, each problem instance is associated with a non-negative integer, called a \textit{parameter}. A parametrized problem $\Pi$ is \textit{fixed-parameter tractable} ($\mathsf{FPT}$) if there is an algorithm that, given an instance $(I,k)$ of $\Pi$, solves it in time $f(k)\cdot |I|^{\mathcal{O}(1)}$ for some computable function $f(\cdot)$. Central to parameterizedcomplexity is the W-hierarchy of complexity classes:
$
\mathsf{FPT} \subseteq \mathsf{W[1]}  \subseteq \mathsf{W[2]} \subseteq \ldots \subseteq \mathsf{XP}.
$

Two instances $I$ and $I'$ (possibly of different problems) are \textit{equivalent} when $I$ is a Yes-instance if and only if $I'$ is a Yes-instance. A \textit{compression} of a parameterized problem $\Pi_1$ into a (possibly non-parameterized) problem $\Pi_2$ is a polynomial-time algorithm that maps each instance $(I,k)$ of $\Pi_1$ to an equivalent instance $I'$ of $\Pi_2$ such that size of $I'$ is bounded by $g(k)$ for some computable function $g(\cdot)$. If $g(\cdot)$ is polynomial, then the problem is said to admit a \textit{polynomial compression}.
A \textit{kernelization algorithm} is a compression where $\Pi_1 = \Pi_2$. Here, the output instance is called a \textit{kernel}.  
Let $\Pi_1$ and $\Pi_2$ be two parameterized problems. A \textit{polynomial parameter transformation} from $\Pi_1$ to $\Pi_2$ is a polynomial time algorithm that, given an instance $(I,k)$ of $\Pi_1$, generates an equivalent instance $(I',k')$ of $\Pi_2$
such that $k' \leq p(k)$, for some polynomial $p(\cdot)$. It is well-known that if $\Pi_1$ does not admit a polynomial compression, then $\Pi_2$ does not admit a polynomial compression~\cite{bookParameterized}. We refer to the books~\cite{bookParameterized,bookKernelization} for  details on parameterized complexity.

\section{Bounding the Cop Number}\label{S:VCBound}
In the following lemma, we give a general upper bound for the cop number, which we use to derive bounds for several graph parameters. 
\begin{lemma}\label{L:boundGeneral}
Let $G$ be a graph and let $U \subseteq V(G)$ be a set of vertices such that for each connected component $H$ of $G[V(G) \setminus U]$, $\mathsf{c}(H) \leq \ell$. Then, $\mathsf{c}(G) \leq \lceil\frac{|U|}{2}\rceil +\ell$.
\end{lemma}
\begin{proof}
We note that this proof uses techniques used to bound $\mathsf{c}(G)$ in terms of $tw(G)$ by Joret et al.~\cite{joret}. Denote $U = \{u_1,\ldots, u_q\}$. Consider isometric paths $P_1, \ldots, P_{\lceil \frac{q}{2}\rceil}$ such that the endpoints of $P_{i}$ are $u_{2i-1}$ and $u_{2i}$. Note that these isometric paths always exist as we assume that the graph is connected. Here, $P_{\lceil \frac{q}{2}\rceil}$ might be a single vertex path containing only the vertex $u_q$. 
Now, we guard each path $P_i$ using a single cop (due to Proposition~\ref{P:aigner}). These $\lceil \frac{q}{2}\rceil$ cops restrict the robber to one  connected component $H$ of $G[V(G)\setminus U]$. Since each of these components is $\ell$-copwin, $ \frac{q}{2} +\ell$ cops have a clear winning strategy in $G$. 
\end{proof}

We know that the classes of star graphs, complete graphs, chordal graphs, and trees are copwin~\cite{nowakowski}.  These bounds, along with Lemma~\ref{L:boundGeneral}, implies the  following theorem.
\begin{theorem}\label{th:ABound}
Let $G$ be a graph and $t =\min\{\mathsf{cvd}, \mathsf{dts}\}$. Then, $\mathsf{c}(G) \leq \frac{t}{2}+1$.
\end{theorem}

\subsection{Bounding Cop Number by $\mathsf{vcn}$:}
Let $U$ be a vertex cover of size $t$ in $G$ and $I$ be the independent set $V(G) \setminus U$. Lemma~\ref{L:boundGeneral} implies that $\mathsf{c}(G) \leq \lceil \frac{t}{2}\rceil +1$. In this section, we improve this bound. First, we provide the following reduction rules. 

\begin{RR}[RR\ref{R:B1}]\label{R:B1}
If there is a vertex $v \in I$ such that $|N(v)| \geq 3$, then place a cop at $v$ and delete $N[v]$.
\end{RR}

\begin{RR}[RR\ref{R:B2}]\label{R:B2}
If there is a vertex $v \in U$ such that $|N[v] \cap U| \geq 3$, then place a cop at $v$ and delete $N[v]$.

\end{RR}

\begin{RR}[RR\ref{R:B3}]\label{R:B3}
If there is an isometric path $P$ such that $P$ contains at least three vertices from $U$, then guard $P$ using one cop and delete $V(P)$ (see Proposition~\ref{P:aigner}).
\end{RR}





We remark that RR\ref{R:B1} and RR\ref{R:B2} can be merged, but we prefer to keep them separate to ease the presentation. Moreover, we note the following.
\begin{note}
In the application of reduction rules RR\ref{R:B1}-RR\ref{R:B3}, whenever a set of vertices $X \subseteq V(G)$ is deleted by the application of rules RR\ref{R:B1}-RR\ref{R:B3}, it implies that each vertex $x \in X$ is being guarded by some cop, and hence, is not accessible to $\R$. We do not actually delete the vertices, and this deletion part is just for the sake of analysis. Hence, from the cop player's perspective, the graph remains connected.
\end{note}

Second, we have the following lemma concerning the structure of subgraphs accessible to $\R$ after an exhaustive application of rules RR\ref{R:B1}-RR\ref{R:B3}.
\begin{lemma}\label{O:shareVertex}
Let $H$ be a connected component of $G$ where rules RR\ref{R:B1}-RR\ref{R:B3} cannot be applied anymore. Then, for every two distinct vertices $x,y \in V(H) \cap U$, either $xy \in E(G)$ or there exists a vertex $w \in I$ such that $xw \in E(G)$ and $yw \in E(G)$.    
\end{lemma}
\begin{proof}
For contradiction, let us assume that there exist two distinct vertices $x,y \in V(H) \cap U$ such that $xy \notin E(G)$ and there does not exist a vertex $w \in I$ such that $xw \in E(G)$ and $yw \in E(G)$. Since $x$ and $y$ are part of the connected component $H$, there exists an $x,y$-path. Let $P$ be an isometric $x,y$-path. 

Let $P = x, v_1, \ldots , v_\ell, y$. Since vertices in $I$ form an independent set and $\ell \geq 2$, the vertices $v_1, \ldots, v_\ell$ cannot be all from $I$. So, there exists at least one $v_i$, for $i \in [\ell]$, such that $v_i \in U$. Thus, $P$ contains at least three vertices from $U$, and $P$ is an isometric path. Therefore, we can apply RR\ref{R:B3}, and hence, we reach a contradiction.   
\end{proof}

Next, we argue that, after an exhaustive application of rules RR\ref{R:B1}-RR\ref{R:B3}, the cop number of each connected component accessible to $\R$ is bounded. We have the following lemma.

\begin{lemma}\label{L:VCbound}
Once we cannot apply rules RR\ref{R:B1}-RR\ref{R:B3} anymore, let the robber be in a connected component $H$. Then, $c(H) \leq 2$.
\end{lemma}
\begin{proof}
We present a winning strategy for two cops. If $H$ contains at most two vertices from $U$, then the cops have a winning strategy by placing a cop on each of these vertices. Hence, we assume there exist at least three vertices in $H$ from $U$. Let $x$ and $y$  be two distinct vertices of $H$ from $U$. Then, we place a cop on each of these vertices. Denote the cops by $\C_1$ and $\C_2$. We consider the two cases as follows.

\smallskip \noindent {\bf Case 1:} If $\R$ is on a vertex in $w \in I$, then due to reduction rule RR\ref{R:B1}, it can have at most two neighbors in $U$. Let them be $u$ and $v$. Now, due to Lemma~\ref{O:shareVertex}, the cops can move to vertices such that one of them, say $x'$, dominates the vertex $u$ and the other, say $y'$, dominates the vertex $v$. See Figure~\ref{fig:VC1} for reference. So, the cops move to the vertices $x'$ and $y'$. This restricts $\R$ to stay on its current vertex $w$ in $I$ (else it is captured in the next move of the cops). Now, in the next move of the cops, they move to the vertices $u$ and $v$. Again, this restricts $\R$ to stay on the vertex $w$ (else it is indeed captured). Finally, in the next move of the cops, the cops capture $\R$. 

\begin{figure}
    \centering
    \includegraphics[scale = 0.90]{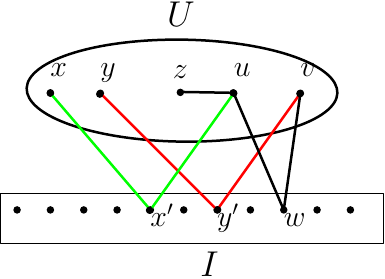}
    \caption{Illustration for the proof of Lemma~\ref{L:VCbound}.}
    \label{fig:VC1}
\end{figure}

\smallskip \noindent {\bf Case 2:} If $\R$ is on a vertex in $u \in U$, then $\C_1$ can move to a vertex in $I$, say $x'$, to attack $\R$ (due to Lemma~\ref{O:shareVertex}). This forces $\R$ to either move to a vertex $w \in I$ or to a vertex $z \in U$. Accordingly, we consider two sub-cases.
\begin{enumerate}
    \item If $\R$ moves to a vertex $ w \in I$, then note that $w$ can have at most two neighbors in $U$ (due to RR1), and one of them is $u$ (being attacked by $\C_1$). Let the other neighbor of $w$ be $v$. Now, $\C_2$ can move to a vertex such that it attacks $v$ (due to Lemma~\ref{O:shareVertex}). This game state is identical to case 1. Hence, the cops can capture the robber in two rounds.
    
    \item If $\R$ moves to a vertex $z \in U$, then $\C_1$ moves to $u$. This forces $\R$ to move to a vertex in $I$ since $u$ can have only one neighbor in $U$ (due to RR\ref{R:B2}), and that is occupied by $C_1$, with both cops being in $U$. This game state is again identical to case 1, and thus the cops win in at most two rounds.
\end{enumerate}

This completes our proof.
\end{proof}

Finally, we have the following theorem.

\VCBound*
\begin{proof}
The correctness of this theorem follows from  Lemma~\ref{L:VCbound} and the fact that using each cop in the reduction rules RR\ref{R:B1}, RR\ref{R:B2}, and RR\ref{R:B3}, we remove at least three vertices from $U$. If we can apply these rules $\frac{t}{3}$ times, $\R$ gets restricted to a vertex in $I$, and thus one additional cop can capture $\R$. Else, when we apply these rules at most $\frac{t}{3}-1$ times, we then need two additional cops (by Lemma~\ref{L:VCbound}), that is, we overall need at most $\frac{t}{3} +1$ cops to ensure capture. 
\end{proof}

We note here that a similar technique will fail if we try to ``remove'' four vertices in each reduction rule. More precisely, if we have the following reduction rules, then we might not get a graph with a bounded (by a constant independent of the input) cop number.

\begin{RR}[RR\ref{R:B4}]\label{R:B4}
If there is a vertex $v \in I$ such that $|N(v)| \geq 4$, then place a cop at $v$ and delete $N[v]$.
\end{RR}

\begin{RR}[RR\ref{R:B5}]\label{R:B5}
If there is a vertex $v \in U$ such that $|N[v] \cap U| \geq 4$, then place a cop at $v$ and delete $N[v]$.

\end{RR}

\begin{RR}[RR\ref{R:B6}]\label{R:B6}
If there is an isometric path $P$ such that $P$ contains at least four vertices from $U$, then guard $U$ using one cop and delete $V(P)$ (see Proposition~\ref{P:aigner}).
\end{RR}

We have the following claim.
\begin{lemma}\label{L:Best}
For every $k \in \mathbb{N}$, there exists a graph $G$ with a vertex cover $U$ and independent set $I = V(G) \setminus U$, such that we cannot apply the rules RR\ref{R:B4}-RR\ref{R:B6}, and $\mathsf{c}(G)>k$. 
\end{lemma}
\begin{proof}

Bonato and Burgess~\cite{designs} proved that for every $k$, there exists a diameter-2 graph $H$ such that $c(H) \geq k$. Let $H$ be a diameter-2 graph such that $c(H) \geq k$.
Joret et al.~\cite{joret} showed that subdividing each edge of a graph an equal number of times does not reduce the cop number. So, we subdivide each edge of $H$ to get the graph $G$ such that $\mathsf{c}(G) \geq k$.  Now, we can put the original vertices in the vertex cover $U$, and the newly introduced vertices in the independent set $I$. We cannot apply any of the rules RR\ref{R:B4} (because each vertex in $I$ has degree exactly 2), RR\ref{R:B5} (because $U$ is an independent set), and RR\ref{R:B6} (since any isometric path in $G$ containing more than three vertices of $U$ will contradict the fact that $H$ is a diameter-2 graph). 
Hence, $G$ is a graph that satisfies the conditions of our lemma.
\end{proof}

\subsection{Bounding the Cop Number for Variants}
Here we extend the result of Theorem~\ref{th:VCbound} to several variations of the game. In particular, we prove the following result.

\begin{lemma}\label{L:variationBound}
Let $G$ be a graph with a vertex cover $U$ of size $t$. Then,
\begin{enumerate}
    \item $\mathsf{c}_{lazy} \leq \frac{t}{2} +1$.
    \item $\mathsf{c}_{attack} \leq \frac{t}{2} +1$.
\end{enumerate}
\end{lemma}
\begin{proof}
Let $I$ be the independent set $V(G) \setminus U$. First, we note here that in \CAR, one cop cannot ensure guarding of an isometric paths~\cite{bonatocar}, and in \LCR, multiple cops, say, $\ell$ cops, cannot ensure guarding $\ell$ paths simultaneously. (This is evident from the fact that there exists a planar graph $G$ with $\mathsf{c}_{lazy}(G) \geq 4$~\cite{lazyplanar}.)  Therefore, reduction rules RR\ref{R:B1}-RR\ref{R:B3} will not directly imply an upper bound on the respective cop numbers here.  So, we have the following reduction rules: 

\begin{RR}[RR\ref{R:BV1}]\label{R:BV1}
If there is a vertex $v \in I$ such that $N(v) > 1$, then place a cop at $v$ and delete $N[v]$.
\end{RR}

\begin{RR}[RR\ref{R:BV2}]\label{R:BV2}
If there is a vertex $v \in U$ such that $N_U[v] > 1$, then place a cop at $v$ and delete $N[v]$.
\end{RR}

Observe that after an exhaustive application of reduction rules RR\ref{R:BV1} and RR\ref{R:BV2}, we are left with a collection of stars, each of which has its center vertex in $U$.

In the case of \LCR, we can easily apply rules RR\ref{R:BV1} and RR\ref{R:BV2}, since cops do not move once placed according to an application of reduction rules RR\ref{R:BV1} and RR\ref{R:BV2}, except for when they move to capture $\R$. Finally, $\R$ is restricted to a star, and one extra lazy cop can move and capture $\R$. 

In the case of \CAR, all cops start at the same vertex. Whenever the cop player wants to station one of the cops at a vertex $v$ according to rules RR\ref{R:BV1} and RR\ref{R:BV2}, all of the cops that are not stationed yet move together to the vertex $v$ (to avoid getting attacked). Note that once a cop is stationed at a vertex $u$, the cop never moves and hence can never be attacked (because if $\R$ wants to attack a cop at vertex $v$, it has to reach a vertex in $N(v)$ in the previous round, and now the cop at $v$ can move and capture $\R$). Once we cannot apply rules RR\ref{R:BV1} and RR\ref{R:BV2} anymore, $\R$ is restricted to a star. At this point, if there are at least two unstationed cops, then these two cops can move to capture $\R$. Else, let $v$ be the last vertex where we stationed a cop. Since at this point we have stationed all but one cop ($\frac{t}{2}$ cops stationed), observe that for each vertex $x\in U$, there is a cop in $N[x]$, and therefore, $\R$ is restricted to one vertex, say, $u$, of $I$. Now, $\R$ can only attack a cop if it is at a vertex in $N(u)$ (and $N(u)\subseteq U$). Finally, the only unstationed cop, say, $\C$, moves to a vertex in $N(u)$ in a finite number of steps (at this point $\R$ cannot attack $\C$ without getting captured as $\C$ is on a vertex in $U$), and captures $\R$ in the next round.  

The bound on the cop numbers follow from the fact that in each reduction rule, we remove at least two vertices from $U$ and place only one cop.
\end{proof}

We have the following straightforward observation concerning the bounds on the cop number for the remaining variants.
\begin{observation}\label{O:VCVariantBound}
Let $t$ be the $\mathsf{vcn}$ of a graph $G$. Then, $\mathsf{c}_{active}(G) \leq t$, $\mathsf{c}_{surround}(G) \leq t$, $\mathsf{c}_s(G) \leq t$ (for any value of $s$), and for a strongly connected orientation $\overrightarrow{G}$ of $G$, $\mathsf{c}(\overrightarrow{G}) \leq t$.
\end{observation}

We remark that the cop number for an oriented graph $\overrightarrow{G}$ (with underlying graph $G$) that is not strongly connected can be arbitrarily  larger than the $\mathsf{vcn}$ of $G$. To see this, consider a vertex cover $U$ of size $t$ in $G$. Next, we add $\ell$ vertices $v_1, \ldots, v_\ell$ such that each vertex $v_i$, for $i \in [\ell]$,  has only outgoing edges to vertices in $U$.  Now, if we do not place a cop on some  $v_j$, for $j \in [\ell]$, then $\R$ can start at $v_j$ and cops can never capture $\R$. Hence, $\mathsf{c}(\overrightarrow{G}) \geq \ell$.

The proof of Theorem~\ref{T:variationBound} directly follows from Lemma~\ref{L:variationBound} and Observation~\ref{O:VCVariantBound}.

\section{Kernelization Algorithms}\label{S:kernel}
In this section, we provide kernelization algorithms for \CR and its variants.

\subsection{Exponential Kernel for \CR by $\mathsf{vcn}$:}
Let $G$ be a graph where a vertex cover $U$ of size $t$ is given. If no such vertex cover is given, then we can compute a vertex cover $U$ of size $t\leq 2\cdot \mathsf{vc}(G)$ using a polynomial-time approximation algorithm~\cite{bookApprox}. Then, the vertices in $V(G) \setminus U$ form an independent set $I$ of size $n-t$. Recall that the question is whether $G$ is $k$-copwin. 

Our kernelization algorithm is based on the exhaustive application of the following reduction rules.

\begin{RR}[RR\ref{R:KV1}]\label{R:KV1}
If $k \geq \frac{t}{3}+1$, then answer positively.
\end{RR}

\begin{RR}[RR\ref{R:KV2}]\label{R:KV2}
If $k = 1$, then apply an $\mathcal{O}(n^3)$ time algorithm (Proposition~\ref{P:XP}) to check whether $G$ is copwin.

\end{RR}

\begin{RR}[RR\ref{R:KV3}]\label{R:KV3}
If there are two distinct vertices $u,v \in I$ such that $N(u) \subseteq N(v)$, then delete $u$.
\end{RR}

The safeness of rule RR\ref{R:KV1} follows from Theorem~\ref{th:VCbound}. For the safeness of rule RR\ref{R:KV3}, we have the following lemma. We note that Lemma~\ref{L:VCkernel} can also be derived from \cite[Corollary 3.3]{berarducci}, but we give a self-contained proof for the sake of completeness.

\begin{lemma}\label{L:VCkernel}
Let $u$ and $v$ be two distinct vertices of $G$ such that $N(u) \subseteq N(v)$. Consider the subgraph $H$ of $G$ induced by $V(G)\setminus \{u\}$. Let $k\geq 2$. Then, $G$ is $k$-copwin if and only if $H$ is $k$-copwin.
\end{lemma}
\begin{proof}
First, we show that if $G$ is $k$-copwin, then $H$ is $k$-copwin. For the graph $H$, the $k$ cops borrow the winning strategy that they have for $G$, with the only difference that whenever a cop has to move to the vertex $u$ in $G$, it moves to $v$ (in $H$) instead. Since $N(u) \subseteq N(v)$, the cop can make the next move as it does in the winning cop strategy for $G$. Note that using this strategy, the cops can capture $\R$ if $\R$ is restricted to $V(H)$ in $G$. Therefore, using this strategy, $k$ cops will capture $\R$ in $H$ as well.

Second, we show that if $H$ is $k$-copwin, then $G$ is is $k$-copwin. Here, for each vertex $x \neq u$ of $G$, we define $I(x) = x$, and for $u$, we define $I(u)= v$. Observe that for each $x \in V(G)$, $I(x)$ is restricted to $H$ and if $xy \in E(G)$, then $I(x)I(y) \in E(H)$. Therefore, every valid move of a player from a vertex $x$ to $y$ in $G$ can be translated to a valid move from $I(x)$ to $I(y)$ in $H$. Now, the cops have the following strategy. If the robber is on a vertex $x$, the cops consider the \textit{image} of the robber on the vertex $I(x)$. Since the robber's image is restricted to $H$, the cops can use the winning strategy for $H$ to capture the image of the robber in $G$. Once the image is captured, if the robber is not on the vertex $ u$, then the robber is also captured. Otherwise, the robber is on the vertex $u$, and at least one cop is on $v$. See Figure~\ref{fig:VCLemma} for an illustration. So, one cop, say $\mathcal{C}_1$, stays on $v$ and this prevents the robber from ever leaving $u$. Indeed this follows because $N(u) \subseteq N(v)$, and so, if $\R$ ever leaves $u$, it will be captured by $\mathcal{C}_1$ in the next cop move. Finally, since $k>1$, some other cop, say $\mathcal{C}_2$, can use a finite number of moves to reach $u$ and capture the robber. 

This completes our proof.
\end{proof}

\begin{figure}
    \centering
    \includegraphics[scale = 0.7]{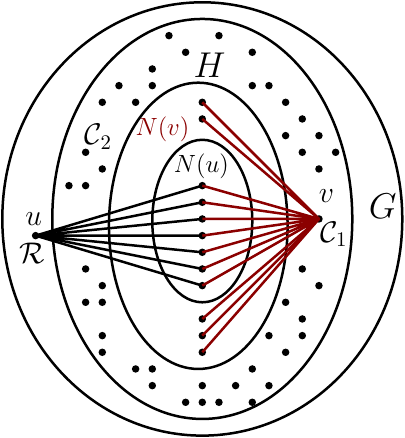}
    \caption{Illustration for Lemma~\ref{L:VCkernel}. Here, $\R$ is at vertex $u$ and $\C_1$ is at vertex $v$.}
    \label{fig:VCLemma}
\end{figure}

Note that the requirement for $k\geq 2$ in Lemma~\ref{L:VCkernel} is crucial. It might so happen that we can get an $H$ such that $c(H)=1$, but $\mathsf{c}(G)>1$. To see this, consider the example of $C_4$, where any two diagonal (i.e., non-adjacent) vertices satisfy the property in Rule RR9, and if we remove one of them, the cop number reduces from 2 to 1. However, this does not harm our algorithm because if we are given $k= 1$, then RR\ref{R:KV2} is applied (before RR\ref{R:KV3}).

Two sets $A$ and $B$ are \textit{incomparable} if neither $A\subseteq B$ nor $B\subseteq A$. We shall use the following proposition that follows from Sperner's Theorem and Stirling's approximation. 
\begin{proposition}\label{P:comb}
Let $X$ be a set of cardinality N. Moreover, let $Y$ be a set of subsets of $X$ such that for each $a,b \in Y$, $a$ and $b$ are incomparable. Then, $|Y| \leq \frac{2^N}{\sqrt{N}}$.
\end{proposition}

Once we cannot apply RR\ref{R:KV1}-RR\ref{R:KV3} anymore, we claim that the size of the reduced graph $G'$ is bounded by a function of $t$. Let $U' = U \cap V(G')$ and $I' = I \cap V(G')$. Clearly, $|U'| \leq t$. Now, each vertex $u \in I'$ is associated with a neighborhood $N(u)$ such that $N(u) \subseteq U'$. Moreover, for any two vertices $u,v \in I'$, the sets $N(u)$ and $N(v)$ are incomparable. Hence, due to Proposition~\ref{P:comb}, $|I'| \leq \frac{2^t}{\sqrt{t}}$, and therefore, $|V(G')| \leq t+\frac{2^t}{\sqrt{t}}$,  which proves the following theorem.
\VCKernel*

Now, we can apply the XP-time algorithm (Proposition~\ref{P:XP}) for \CR on our kernel. Since $k \leq \frac{t}{3}$, the running time we get is exponential only in $t$ and polynomial in $n$. Specifically, the running time of the algorithm is $ t\cdot \big{(}t+\frac{2^t}{\sqrt{t}}\big{)}^{\frac{t}{3}+2} \cdot n^{\mathcal{O}(1)}$. Moreover, if a vertex cover $U$ of size $t = \mathsf{vc}(G)$ is not given, then we can compute one in time $1.2738^t\cdot n^{\mathcal{O}(1)}$~\cite{fastestVC}. Thus, we have the following corollary.

\VCFPT*

\subsection{Exponential Kernel for \CR by \cvd:}\label{S:clique} 
To get a kernel for \CR parameterized by $\mathsf{cvd}$, we employ techniques similar to the ones we used to get a kernel for \CR parameterized by $\mathsf{vcn}$. Let $U$ be a cluster vertex deletion set of size $t$. Let $S = V(G)\setminus U$, and  $C_1, \ldots, C_\ell$ be the set of disjoint cliques that form the graph $G[S]$. Since $\mathsf{c}(G)\leq \frac{t}{2}+1$ (Theorem~\ref{th:ABound}), we have the following reduction rule.

\begin{RR}[RR\ref{R:KC1}]\label{R:KC1}
If $k \geq \frac{t}{2}+1$, then report Yes-instance.
\end{RR}

Next, we have the following lemma. 

\begin{lemma}\label{l:cliquesize}
Let $u$ and $v$ be vertices of some clique $C$ of $G[S]$. If $N_U(u) \subseteq N_U(v)$, then $\mathsf{c}(G) = \mathsf{c}(G\setminus  \{u\})$.
\end{lemma}
\begin{proof}
First we observe that, since $u$ and $v$ are part of the same clique $C$, $N[u] \subseteq N[v]$. Then, the proof of this lemma follows from the proof of Lemma~\ref{L:VCkernel}. We remark that this proof also follows from the idea of retracts used in the \CR literature~\cite{berarducci,nowakowski}. Additionally, we remark that, here, $\mathsf{c}(G)$ need not be necessarily greater than 1. To see this, consider the situation when $\R$ is at $u$ and a cop, say, $\C_1$, is at $v$. Now, $\R$ cannot move to a vertex in $U$ since $N_U(u) \subseteq N_U(v)$ and $\R$ cannot stay on a vertex in $C$ since $v$ is a part of $C$. Thus, $\R$ gets captured in the next move by $\C_1$.
\end{proof}

Hence, we can apply the following reduction rule, whose safeness was proved by Lemma~\ref{l:cliquesize}.


\begin{RR}[RR\ref{R:KC2}]\label{R:KC2}
Let $u$ and $v$ be vertices of some clique $C \in G[S]$ such that $N[u] \subseteq N[v]$. Then, delete $u$.
\end{RR}

Once we cannot apply reduction rule RR\ref{R:KC2} anymore, the size of each clique in $G[S]$ is at most $\frac{2^t}{\sqrt{t}}$ (due to Proposition~\ref{P:comb}).

Similarly to Lemma~\ref{l:cliquesize}, we have the following lemma.

\begin{lemma}\label{L:cliqueType}
Let $C_i$ and $C_j$ be two cliques in $G[S]$ such that for each vertex $u \in V(C_i)$, there exists a vertex $v \in V(C_j)$ such that $N_U(u) \subseteq N_U(v)$. Then, $k >1 $ cops have a winning strategy in $G$ if and only if they have a winning strategy in $G[V(G) \setminus V(C_i)]$.
\end{lemma}
\begin{proof}
The proof idea here is similar to the proof idea of Lemma~\ref{L:VCkernel}. Let $H= G[V(G) \setminus V(C_i)]$. Here, we will just prove that if $k$ cops have a winning strategy in $H$, then $k$ cops have a winning strategy in $G$. (The proof of the reverse direction is rather easy to see, combining arguments from Lemma~\ref{L:VCkernel} and the arguments we present in the rest of this proof).

Let $k\geq 2$ cops have a winning strategy in $H$. Similarly to Lemma~\ref{L:VCkernel}, for each vertex $x \in V(G) \setminus V(C_i)$, we define $I(x) = x$, and for each vertex $u\in V(C_i)$, we have a vertex $v\in V(C_j)$ such that $N_U(u) \subseteq N_U(v)$, and we define $I(u)=v$. (Note that there might be multiple choices for $v$ here. We can choose any such vertex.)

Observe that for each vertex $x\in V(G)$, $I(x)$ is restricted to $H$. Moreover, if $xy \in E(G)$, then $I(x)I(y) \in E(H)$ for the following reasons. If $x,y \in V(G) \setminus V(C_i)$, then it is obvious. Else, if $x,y \in V(C_i)$, then observe that $I(x)$ and $I(y)$ are part of some clique $C_j$, and $N_U(x) \subseteq N_U(I(x))$ and $N_U(y) \subseteq N_U(I(y))$. Hence, in this case, if $xy \in E(G)$, then $I(x)I(y) \in E(H)$. Finally, assume without loss of generality that $x\in V(C_i)$ and $y\in V(G)\setminus V(C_i)$. In this case, $xy\in E(G)$ only if $y\in U$. Since $N_U(x) \subseteq N_U(I(x))$, $I(x)I(y) \in E(H)$. Thus, if $xy \in E(G)$, then $I(x)I(y) \in E(H)$. Therefore, every valid move of a player from a vertex $x$ to a vertex $y$ in $G$ can be translated to a move from $I(x)$ to $I(y)$ in $H$.

Now, cops play their winning strategy in $H$ with the following consideration: When the robber is at a vertex $x$ in $G$, the cops consider the \textit{image} of the robber at vertex $I(x)$ in $G$. Since the robber's image is restricted to the vertices of $H$, the cops can use a winning strategy from $H$ to capture the image of the robber in $G$. Once the image is captured, if the robber is at a vertex $x \notin V(C_i)$, then the robber is also captured. Otherwise, the robber is at a vertex $x\in V(C_i)$, and one of the cops is at vertex $I(x)$ in $C_j$. Now, observe that the robber cannot immediately move to a vertex in $U$. Anyhow, the robber can move to some other vertex $y \in V(C_i)$, and in this case, the cop at vertex $I(x)$ can move to vertex $I(y) \in V(C_j)$. This way, the cop occupying the robber's image can prevent the robber from ever leaving $C_i$. Since $k\geq 2$, some other cop can move to capture the robber in $C_i$ (as cliques are copwin). This completes our proof. 
\end{proof}

Thus, we can apply the following reduction rule, whose safeness was proved by Lemma~\ref{L:cliqueType}.

\begin{RR}[RR\ref{R:KC3}]\label{R:KC3}
Let $C_i$ and $C_j$ be two cliques in $G[S]$ such that for each vertex $u \in V(C_i)$, there exists a vertex $v \in V(C_j)$ such that $N_U(u) \subseteq N_U(v)$. Then, delete $V(C_i)$.
\end{RR}

Finally, we use the following lemma to bound the size of the desired kernel from Theorem~\ref{T:vertexClique}.
\begin{lemma}\label{L:boundClq}
    After an exhaustive application of RR\ref{R:KC1}-RR\ref{R:KC3}, the size of the reduced graph is at most $2^{2^t + \sqrt{t}}$.
\end{lemma}
\begin{proof}    
Once we cannot apply the reduction rules RR\ref{R:KC2} and RR\ref{R:KC3}, due to Proposition~\ref{P:comb}, each clique can have at most $\frac{2^t}{\sqrt{t}}$ vertices. Moreover, the total number of cliques possible is at most $\frac{2^{\frac{2^t}{\sqrt{t}}}}{\sqrt{\frac{2^t}{\sqrt{t}}}}$ (due to Proposition~\ref{P:comb}). Thus, the total number of vertices in the reduced graph is at most $2^{2^t + \sqrt{t}}$.
\end{proof}


Since $k \leq \frac{t}{2}+1$ (by Reduction Rule RR\ref{R:KC1}), applying the XP-algorithm for \CR from Proposition~\ref{P:XP} to the kernel in Theorem~\ref{T:vertexClique} gives us the following corollary.

\begin{corollary}
\CR is \FPT parameterized by \cvd. Specifically, it is solvable in $(\mathsf{cvd}+2^{2^\mathsf{cvd} + \sqrt{\mathsf{cvd}}})^{\frac{\mathsf{cvd}}{2}+2}\cdot n^{\mathcal{O}(1)}$ time. 
\end{corollary}

\subsection{Exponential Kernel for \CR by \dts}
Using the ideas we presented in Section~\ref{S:clique}, we can also get a kernel for \CR with respect to deletion to stars number. Let $U$ be a deletion to stars vertex set of size $t$. Also, let $S = V(G) \setminus U$, and let $X_1, \ldots X_\ell$ be the stars in the graph $G[S]$. Specifically, we have the following reduction rules along with reduction rule RR\ref{R:KC1}. 

\begin{RR}[RR\ref{R:KC4}]\label{R:KC4}
Let $u$ and $v$ be two leaf vertices of some star $X$ in $G[S]$ such that $N_U(u) \subseteq N_U(v)$. Then, delete $u$.
\end{RR}

\begin{RR}[RR\ref{R:KC5}]\label{R:KC5}
Let $X$ and $Y$ be two stars in $G[S]$ such that $V(X) = x, x_1, \ldots, x_p$ and $V(Y) = y, y_1, \ldots, y_q$, where $x$ and $y$ are center vertices of $X$ and $Y$, respectively. If $N_U(x)\subseteq N_U(y)$ and for each vertex $x_i$ (for $i\in [p]$), there is a vertex $y_j$ (for $j\in [q]$) such that $N_U(x_i) \subseteq N_U(y_j)$, then delete $X$.
\end{RR}

The safeness of RR\ref{R:KC1} follows from Theorem~\ref{th:ABound}. We have the following lemma, which establishes that reduction rules RR\ref{R:KC4} and RR\ref{R:KC5} are safe.

\begin{lemma}\label{L:dtsnSafe}
Assuming $k>1$, reduction rules RR\ref{R:KC4} and RR\ref{R:KC5} are safe.
\end{lemma}
\begin{proof}
To prove that rule RR\ref{R:KC4} is safe, it suffices to observe that for leaf vertices $u$ and $v$ of some star $X \in S$, if $N_U(u) \subseteq N_U(v)$, then $N(u) \subseteq N(v)$ in $G$. Indeed, the rest of the proof follows directly from the proof of Lemma~\ref{L:VCkernel}.

Next, we give a proof idea for the safeness of rule RR\ref{R:KC5}. Here, we just define the  function of the image of the robber, and the rest of the proof is similar to the proofs of Lemmas~\ref{L:VCkernel} and~\ref{L:cliqueType}.  For each vertex $u \notin V(X)$, $I(u) = u$. For each $x_i$, $I(x_i) = y_j$ such that $N_U(x_i)\subseteq N_U(y_j)$ (there might be multiple choices for $y_j$ and we can choose any one of them), and $I(x) = y$.
\end{proof}

Now, we claim that once we cannot apply rules RR\ref{R:KC4} and RR\ref{R:KC5} anymore, the size of the graph is bounded by a function of $t$. First, we note that the size of each star is at most $\frac{2^t}{\sqrt{t} }+1$ (by Proposition~\ref{P:comb}). Let $X$ and $Y$ be two stars in $G[S]$ such that $x$ and $y$ are the center vertices of $X$ and $Y$, respectively. We say that $X$ and $Y$ have the same \textit{neighbourhood type} if  $N_U(x) = N_U(y)$. Second, it is easy to see that there can be at most $2^t$ neighbourhood types. Next, we bound the number of stars in each neighbourhood type. Let $S_1, \ldots, S_z$ be the stars having the same neighbourhood type, and let $v_i$ be the center vertex of star $S_i$. For each star $S_i$, for $i\in [z]$, let $\mathcal{S}_i = \{N(v): v\in V(S_i)\setminus \{v_i\} \}$. Since we have applied reduction rule RR\ref{R:KC4} exhaustively, we know that for each $A\in \mathcal{S}_i$, $A=N(v)$ for a unique vertex $v\in V(S_i)\setminus \{v_i\}$. Observe that each $S_i$ is a subset of the power set of $U$ and the power set of $U$ has size $2^t$. Moreover, since we have applied reduction rule RR\ref{R:KC5} exhaustively, we know that for any $i,j \in [z]$, neither $\mathcal{S}_i \subseteq \mathcal{S}_j$ nor $\mathcal{S}_j \subseteq \mathcal{S}_i$. Hence, due to Proposition~\ref{P:comb}, $z \leq \frac{2^{2^t}}{\sqrt{2^t}}$. 
Therefore, the size of the reduced graph can be at most $\frac{2^{2^t}}{\sqrt{2^t}}\cdot 2^t \cdot (\frac{2^t}{\sqrt{t}} +1) $. Thus, we have the desired kernel from Theorem~\ref{T:vertexClique}.


Since $k \leq \frac{t}{2}+1$ (by reduction rule RR\ref{R:KC1}), applying the XP-algorithm for \CR from Proposition~\ref{P:XP} to the kernel in Theorem~\ref{T:vertexClique} gives us the following  corollary.

\begin{corollary}
\CR is \FPT parameterized by \dts. Specifically, it is solvable in $(\mathsf{dts}+2^{2^\mathsf{dts} + \mathsf{dts}^{1.5}})^{\frac{\mathsf{dts}}{2}+2}\cdot n^{\mathcal{O}(1)}$ time. 
\end{corollary}

\subsection{Exponential Kernels for Different Variants}

Here, we extend the result of Theorem~\ref{th:Kernel} to several variations of the game. We have the following results.

\subsubsection{\LCR and \CAR:}
First, we prove the following lemma.
\begin{lemma}\label{L:VCkernelVariants}

Let $u$ and $v$ be two distinct vertices of $G$ such that $N(u) \subseteq N(v)$. Consider the graph $H$ induced by $V(G)\setminus \{u\}$. Then for $k>1$ and for $x\in \{lazy,attack\}$, $\mathsf{c}_x(G) \leq k$ if and only if $\mathsf{c}_x(H) \leq k$. 
\end{lemma}
\begin{proof}
The proof for the forward direction ($c_x(G)\leq k$ implies $c_x(H)\leq k$) is easy and follows from the arguments similar to the arguments in the proof of Lemma~\ref{L:VCkernel}. We prove the reverse side ($c_x(H)\leq k$ implies $ c_x(G) \leq k$) for both  the variants below. Moreover, similarly to the proof of Lemma~\ref{L:VCkernel}, we define $I(u) = v$ and $I(x) = x$ when $x\neq u$. Similarly, when $\R$ is at a vertex $x$, we say that the \textit{image} of $\R$ is at vertex $I(x)$. (Note that the image of $\R$ is restricted to $H$.) In both of the variants, the cops will play in $G$ to capture the image of the robber using the winning strategy of $H$. 

In \LCR, the cops begin by capturing the image of $\R$ in $G$. If $\R$ is at a vertex $x \neq u$, then $\R$ is captured. If $\R$ is at vertex $u$, then observe that there is a cop, say, $\C$, at $v$ that has captured the image of $\R$. Now, $\C$ ensures that $\R$ cannot move, and some other lazy cop can move to capture $\R$ in a finite number of rounds.

In \CAR, the main observation is that if the cops can capture $\R$ in $H$, they can capture the image of $\R$ in $G$ without getting attacked by $\R$. If $\R$ is at a vertex $x\neq u$ when the image of $\R$ is captured, then $\R$ is captured. Otherwise, $\R$ is at $u$, and a cop, say, $\C_1$, is at  vertex $v$. Now another cop, say, $\C_2$, can move to a vertex $w \in N(v)$ (in a finite number of steps) to capture $\R$. If $\R$ attacks $\C_2$ at this point, then note that $\C_1$ can move to capture $\R$ in the next round. If $\R$ does not attack, then $\C_2$ moves to capture $\R$ in the next round.
\end{proof}

Lemma~\ref{L:VCkernelVariants} establishes that  reduction rule RR\ref{R:KV3} is safe for both \CAR and \LCR. Before applying reduction rule RR\ref{R:KV3}, we apply the following reduction rules.

\begin{RR}[RR\ref{R:E1}]\label{R:E1}
If $k \geq \frac{t}{2} +1$, then answer positively (Theorem~\ref{T:variationBound}).
\end{RR}

\begin{RR}[RR\ref{R:E2}]\label{R:E2}
If $k=1$, then apply the ${\mathcal{O}(n^3)}$ time algorithm from Proposition~\ref{P:generalXP}.
\end{RR}

The size of the kernel, by using these reduction rules, is dependent on RR9. Therefore, the proof of the existence of the desired kernel from Theorem~\ref{T:LAtKernel} follows directly.

Moreover, Theorem~\ref{T:LAtKernel}, along with the XP-algorithms from Proposition~\ref{P:generalXP} for these variants, gives the following immediate corollary.

\begin{corollary}
\CAR and \LCR are \FPT parameterized by $\mathsf{vcn}$. Specifically, they are solvable in $\mathsf{vcn}+ \frac{2^\mathsf{vcn}}{\sqrt{\mathsf{vcn}}})^{\frac{\mathsf{vcn}}{2}+2}\cdot n^{\mathcal{O}(1)}$ time. 
\end{corollary}

\subsubsection{\CR on Directed Graphs:}
Next, we consider the game of \CR on oriented graphs. For a directed graph $\overrightarrow{G}$ and a vertex $v\in V(\overrightarrow{G})$, let $N^+(v)$ and $N^-(v)$ denote the set of out-neighbors and in-neighbors of $v$, respectively. We have the following lemma. 
\begin{lemma}\label{L:oriented}
Let $u$ and $v$ be two distinct vertices of a strongly connected directed graph $\overrightarrow{G}$ such that $N^+(u) \subseteq N^+(v)$ and $N^-(u) \subseteq N^-(v)$. Let $\overrightarrow{H}$ be the graph induced by $V(\overrightarrow{G})\setminus \{u\}$. Then, for $k>1$, $k$ cops have a winning strategy in $\overrightarrow{H}$ if and only if $k$ cops have a winning strategy in $\overrightarrow{G}$
\end{lemma}
\begin{proof}
First, observe that $H$ is also strongly connected.

Second, let $k$ cops have a winning strategy in $\overrightarrow{G}$. Then, the cops can use this winning strategy in $\overrightarrow{H}$, with the only difference that whenever a cop, say, $\C$, has to move to $u$ in $\overrightarrow{G}$, $\C$ moves to $v$ in $\overrightarrow{H}$ instead ($\C$ can do so because $N^-(u) \subseteq N^-(v)$). Next, whenever $\C$ has to move to a vertex, say, $w$, from $u$, in the strategy in $G$, then $\C$ can move to $w$ from $v$ also (since $N^+(u) \subseteq N^+(v)$). As $\R$ is restricted to $V(\overrightarrow{H})$ in $\overrightarrow{G}$, cops will capture $\R$ using this strategy in $\overrightarrow{H}$ as well.

Finally, let $k$ cops have a winning strategy in $\overrightarrow{H}$. We use this strategy to get a winning strategy in $\overrightarrow{G}$ using $k$ cops. First, we define $I(x) = x$ for $x\neq u$ and $I(u) = v$. Since $I(x)$ is restricted to $\overrightarrow{H}$, we use the winning strategy in $\overrightarrow{H}$ to capture $I(x)$. At this point if $x \neq u$, then $\R$ is captured. Else, $\R$ is at $u$ and one of the cops, say, $\C_1$, is at $v$. Since $N^+(u) \subseteq N^+(v)$, $\R$ cannot move as long as $\C_1$ occupies $v$. Since $\overrightarrow{G}$ is strongly connected, one of the other cops, say, $\C_2$, can move to $u$ in a finite number of rounds to capture $\R$.
\end{proof}

Let $G$ be a graph with a vertex cover $U$ of size $t$, and let $I = V(G)\setminus U$. Let $\overrightarrow{G}$ be a strongly connected orientation of $G$. We apply the following reduction rules.

\begin{RR}[RR\ref{R:D1}]\label{R:D1}
If $k\geq t$, then answer positively.
\end{RR}

\begin{RR}[RR\ref{R:D2}]\label{R:D2}
If $k=1$, then apply the $\mathcal{O}(n^3)$ time algorithm from Proposition~\ref{P:generalXP} to check whether $\overrightarrow{G}$ is copwin.
\end{RR}

\begin{RR}[RR\ref{R:D3}]\label{R:D3}
If $u$ and $v$ are two distinct vertices in $I$ such that  $N^+(u) \subseteq N^+(v)$ and $N^-(u) \subseteq N^-(v)$, then delete $u$.
\end{RR}

Safeness of reduction rules RR\ref{R:D1} and RR\ref{R:D3} follow from Theorem~\ref{T:variationBound} and Lemma~\ref{L:oriented}, respectively. Now, we argue that once we cannot apply reduction rules RR\ref{R:D1}-RR\ref{R:D3}, the size of $\overrightarrow{G}$ is bounded by a function of $t$. Observe that each vertex $u$ in $I$ has a unique neighbourhood ($N^+(u)\cup N^-(u)$) and there are three choices for a vertex  $v \in U$ to appear in the neighbourhood of a vertex $u \in I$, that is,  either $v \in N^+(u)$, or $v \in N^-(u)$, or $v \notin N^+(u)\cup N^-(u)$. Therefore, the total number of possible vertices in $I$ are at most $3^t$. Thus, applying reduction rules RR\ref{R:D1}-RR\ref{R:D3}, we get the desired kernel from Theorem~\ref{T:LAtKernel}.


Theorem~\ref{T:LAtKernel}, along with rule RR21 and Proposition~\ref{P:generalXP}, gives the following corollary.

\begin{corollary}\label{C:COriended}
\CR on strongly connected directed graphs is \FPT parameterized by the vertex cover number $t$. In particular, it is solvable in $(3^t+t)^{t+1}\cdot n^{\mathcal{O}(1)}$ time.
\end{corollary}

\subsection{General Kernelization}
In this section, we provide a general reduction rule that works for most variants of \CR parameterized by the vertex cover number.
Let $U$ be a vertex cover of size $t$ in $G$ and $I$ be the independent set $V(G) \setminus U$. For each subset $S\subseteq U$, we define the following equivalence class: $\mathcal{C}_S = \{ v \in I \colon  N(v) = S\}$.
Given an instance $((G,k),t)$, we have the following reduction rule.

\begin{RR}[RR\ref{R:G}]\label{R:G}
If there is an equivalence class $\mathcal{C}_S$ such that $|\mathcal{C}_S| >k+1$, then keep only $k+1$ arbitrary vertices from $\mathcal{C}_S$ in $G$, and delete the rest.
\end{RR}

First, we present (informal) intuition why reduction rule RR22 is safe. Since the neighbourhood of each vertex in $\mathcal{C}_S$ is the same, all of these vertices are equivalent with respect to the movement rules in any of the variants discussed. We keep $k+1$ copies of such vertices because, on a robber move, there is at least one vertex that is not occupied by any cop. We refer to such a vertex as a \textit{free vertex}. Note that there might be multiple free vertices. On a robber player's turn, if $\R$ plans to move to a vertex in $\mathcal{C}_S$, it can move to a free vertex. Moreover, if a fast robber wants to use a vertex from $\mathcal{C}_S$ as an intermediate vertex, it can use a free vertex for this purpose as well.  We prove safeness for individual variants later in this section.

Moreover, we have the following lemma that we will use later. 
\begin{lemma}\label{L:generalSize}
Let $G$ be a graph with a vertex cover $U$ of size $t$. After an exhaustive application of reduction rules RR\ref{R:D1} and RR\ref{R:G}, the reduced graph has at most $t+ t\cdot2^t$ vertices.
\end{lemma}
\begin{proof}
There can be at most $2^t$ equivalence classes, and for each equivalence class, we keep at most $k+1$ vertices in $I$. Due to rule RR19, we can assume $k<t$. Thus, size of $I$ is at most $t\cdot 2^t$. The size of $G$ is, therefore, at most $|U| + |I| \leq t+ t\cdot2^t$.
\end{proof}

\subsubsection{\textsc{Generalized CnR}}
In this section, we establish that RR\ref{R:G} is safe for \textsc{Generalized CnR}. We have the following lemma to prove this claim. 

\begin{lemma}\label{L:Gen}
Let $G$ be a graph with a vertex cover $U$ of size $t$. Let $\mathcal{C}_S$ (for $S \subseteq U$) be an equivalence class such that $|\mathcal{C}_S| = \ell >k+1$. Moreover, let $H$ be a subgraph formed by deleting an arbitrary vertex $v$ of $\mathcal{C}_S$ from $G$. Then, $(G,\C_1,\ldots,\C_k, \R)$ is a Yes-instance if and only if $(H,\C_1,\ldots,\C_k, \R)$ is a Yes-instance.
\end{lemma}
\begin{proof}
Let $\mathcal{C}_S = \{v_1, \ldots, v_\ell\}$. Without loss of generality, let us assume that vertices $v_1,\ldots v_{\ell-1}$ belong to the graph $H$, and $v = v_\ell$. Since there are at most $k$ cops in the game and $\ell >k+1$, at least one vertex of $v_1,\ldots, v_{\ell-1}$ is not occupied by any cop. We denote this vertex by $x$ ($x$ is dependent on the position of the cops and may change during the course of the game). Moreover, here we modify the definition of a \textit{safe vertex} slightly: A vertex $y$ is \textit{safe} if it is at a distance at least $\lambda_i+1$ from $\C_i$, for $i\in [k]$. Since each vertex in $\mathcal{C}_S$ has the same neighborhood, observe that either each vertex in $\mathcal{C}_S$ not occupied by a cop is a safe vertex or none of the vertices in $\mathcal{C}_S$ is safe. Moreover, for each vertex $y\in V(G)\setminus \{v\}$, let $I(y) = y$ and $I(v) = x$. Note that for each vertex $u$, $I(u)$ is restricted to vertices of $V(H)$, $N(u)= N(I(u))$, and if $u$ is a safe vertex, then $I(u)$ is also a safe vertex. To ease the presentation, instead of saying that the cops/robber has a winning strategy in $(G,\C_1,\ldots, \C_k,\R)$ (or $(G,\C_1,\ldots, \C_k,\R)$), we will say that the cops/robber has a winning strategy in $G$ (or $H$). 

First, let $\R$ has a winning strategy $\mathcal{S}$ in $G$. To show that $\R$ has a winning strategy in $H$,  we will prove a slightly stronger statement that $\R$ has a winning strategy, say, $\mathcal{S}'$, in $G$ even if $\R$ is restricted to the vertices of $V(H)$ in $G$. We get $\mathcal{S}'$ from $\mathcal{S}$ as follows: If $\R$ has to use a vertex $y$ in $\mathcal{S}$ in some move during the game, it uses $I(y)$ instead.  We first show that $\R$ can safely enter the graph. Let $y$ be the vertex $\R$ enters in the strategy $\mathcal{S}$. Then, $\R$ enters at $I(y)$ in $\mathcal{S}'$. Since $y$ is a safe vertex (as $\mathcal{S}$ is a winning strategy for $\R$), $I(y)$ is also a safe vertex. Hence $\R$ can safely enter a vertex. Now, the only thing to argue is that if $\R$ can move safely from a vertex $y$ to a vertex $z$ in $G$, then it can safely move from vertex $I(y)$ to $I(z)$ in $G$. Let $\R$ moves from $y$ to $z$ using a path $P_1= (y=y_1,\ldots,y_r=z)$, where $r\in [s_R]$, during some move in $\mathcal{S}$. Notice that since $\mathcal{S}$ is a winning strategy, each vertex $y_i$ ($i\in [r]$) is a safe vertex, and hence, each vertex $I(y_i)$ is also a safe vertex. Moreover, since $N(y_i) = N(I(y_i))$, $W=(I(y_1), \ldots, I(y_r))$ is a walk with at most $r$ vertices between $I(y)$ and $I(z)$. (It might not be a path since vertex $x$ may repeat in this walk.) Since the existence of a walk between two vertices implies the existence of a path between these vertices using vertices from a subset of the walk vertices, we have an $I(y),I(z)$-path of length at most $r$ using (safe) vertices from $\{I(y_1),\ldots,I(y_r)\}$. Hence $\R$ can safely move from $I(y)$ to $I(z)$. Thus, $\mathcal{S}'$ is a winning strategy for $\R$ even when $\R$ is restricted to vertices of $V(H)$ in $G$.

In the reverse direction, let $\R$ has a winning strategy in $H$. Then, we show that $\R$ has a winning strategy in $G$ as well. Here, whenever a cop $\C_i$ moves to a vertex $y$, $\R$ assumes its image at the vertex $I(y)$. Observe that $I(y)$ is restricted to $V(H)$ in $G$. Let $y_1,\ldots,y_k$ be the vertices occupied by cops during some instance in the game. Let $F$ be the set of vertices in $V(H)$ that are safe during this turn. Moreover, let $F'$ be the set of the vertices in $V(H)$ that are safe if cops are occupying the vertices $I(y_1), \ldots, I(y_k)$. Then, we have the following claim.

\begin{claim}
$F'\subseteq F$.
\end{claim}\label{C:general}

\begin{proofofclaim}
Targetting contradiction, assume that $y\in F'$ but $y\notin F$. Then, there exists some $i\in [k]$ such that $d(y, y_i) \leq \lambda_i$ but $d(y,I(y_i))>\lambda_i$. If $y_i \neq v$, then this is not possible since for $y_i\neq v$, $I(y_i)=y_i$. Hence, we can assume that $y_i = v$ and $I(y_i) = x$. Since $N(v) = N(x)$, for each vertex $y$ in $V(G)\setminus\{v\}$ (and $y\in F' \subseteq V(H)$), $d(y,x) \leq d(y,v)$, that is, $d(y,I(y_i))\leq d(y,y_i)$, a contradiction.

We note that it might not be true that $F\subseteq F'$, as it might so happen that $F$ contains the vertex $x$, but $F'$ does not.
\end{proofofclaim}

Due to Claim~\ref{C:general}, it is sufficent to show that if $\R$ has a winning strategy in $H$ considering the image of cop $\C_i$ as a cop with the capabilities of $\C_i$, then $\R$ has a winning strategy in $G$. To this end, $\R$ can use its winning strategy from $H$ since images of the cops are restricted to $V(H)$. Thus, $\R$ has a winning strategy in $G$.

Finally, note that, in both directions of proofs, $\R$ moves in $H$ (respectively, in $G$) if and only if $\R$ moves  in $G$ (respectively, in $H$). Hence, if $\R$ is active/flexible in the original strategy, $\R$ is active/flexible in the designed strategy. This completes the proof.
\end{proof}

Observe that $\mathsf{vcn}+1$ cops always have a winning strategy in $G$. Therefore, we have the following theorem as a consequence of Lemma~\ref{L:Gen} and Lemma~\ref{L:generalSize}.

\GenKernel*



Theorem~\ref{T:general} directly implies the existence of the desired kernel for \ACR and \CFR from Theorem~\ref{Th:active}. The existence of the desired kernel for \SCR from Theorem~\ref{Th:active} follows from Lemma~\ref{L:generalSize} and the following lemma, which proves the safeness of RR\ref{R:G} for \SCR.

\begin{lemma}\label{L:ASCR}
Let $G$ be a graph with a vertex cover $U$ of size $t$. Let $\mathcal{C}_S$ (for $S \subseteq U$) be an equivalence class such that $|\mathcal{C}_S| = \ell >k+1$. For a subgraph $H$ formed by deleting $\ell - k - 1$ arbitrary vertices of $\mathcal{C}_S$ from $G$, $\mathsf{c}_{surround}(H) \leq k$ if and only if $\mathsf{c}_{surround}(G) \leq k$.
\end{lemma}
\begin{proof}
Let $\mathcal{C}_S = \{v_1, \ldots, v_\ell\}$. Without loss of generality, let us assume that the vertices $v_1,\ldots v_{k+1}$ belong to the graph $H$ and vertices $v_{k+2}, \ldots, v_\ell$ are deleted. We begin by noting that $\R$ cannot be surrounded at a vertex in $S$ in $G$ (since each vertex in $S$ has at least $k+1$ neighbours). Therefore, throughout the proof, we have the implicit assumption that when $\R$ is surrounded, it is not on a vertex in $S$.

Let $k$ cops have a winning strategy in $G$. Then, to surround $\R$, cops use this strategy with the following changes in $H$. Whenever a cop has to move to a vertex in $\{v_{k+2}, \ldots, v_\ell\}$, it moves to vertex $v_1$ instead. Since all vertices in $\mathcal{C}_S$ have the same neighbourhood, the next move of this cop can be the same as it was in (the winning strategy of) $G$. Note that using this strategy, the cops can surround the robber in $G$ if $\R$ is restricted to $V(H)$ in $G$, and also the moves of cops are restricted to $V(H)$ in $G$. Therefore, the cops can surround $\R$ using this strategy in $H$ as well.

Now, let $k$ cops have a winning strategy in $H$. We use this strategy to surround $\R$ in $G$, in the following manner. Since we have only $k$ cops, during any time in the gameplay, there is at least one vertex in $\{v_1, \ldots, v_{k+1}\}$ that is not occupied by any cop. Let us call this vertex a \textit{free vertex} (there might be multiple free vertices). Again we use the concept of the \textit{image of the robber}. 
For each vertex $x\in V(G)$, if $x\in V(H)$, then we define $I(x) = x$; else, if $x\in \{ v_{k+1}, \ldots v_\ell \}$, then we define $I(x) = y$, where $y$ is a free vertex at that instance. Whenever $\R$ moves to a vertex $x \in V(G)$, we say that the image of the robber moves to $I(x)$. Moreover, we remind that, in this game, although some cop and $\R$ can be at the same vertex, the robber cannot end its move at the same vertex as one of the cops. Cops use this capability to force $\R$ to move from a vertex. Therefore, we also have to argue that whenever cops force $\R$ to move, they force the image of the robber to move as well. To this end, observe that the image of the robber and the robber are on different vertices only if $\R$ is on some vertex $x\in \{ v_{k+1}, \ldots, v_\ell \}$ and the image of the robber is on a free vertex, say, $y$. Notice that if, in the strategy for $H$, $\R$ was occupying $y$ and the cop player wants to force $\R$ to move out of $y$, then it does so by moving a cop, say, $\C$, from a vertex $w\in N(y)$ to $y$. Cop player adapts this strategy in $G$ by moving $\C$ form $w$ to $x$ instead of $w$ to $y$. This move is possible because $N(x)= N(y)$. Thus, $\R$, as well as the image of $\R$, are forced to move as they would have been forced to move in the winning strategy of $k$ cops in $H$.

Hence, the image of $\R$ is restricted to $V(H)$ in $G$ and has to follow the rules of the movement of the robber. 
Thus, the cops will finally surround the image of $\R$ in  $G$. At this point, if $\R$ is on a vertex $ v \in  \{v_{k+2}, \ldots, v_\ell\}$, note that $I(\R)$ is on a vertex  $u \in \{v_1, \ldots, v_{k+1}\}$. Observe that here, if $I(\R)$ is surrounded, then there is a cop on each vertex in $S$, and thus, $\R$ is surrounded as well. If $\R$ was on a vertex in $V(H)\setminus S$ when surrounded, then $I(\R)$ and $\R$ are at the same vertex, and thus, $\R$ is surrounded as well.
\end{proof}

This finishes the proof of Theorem~\ref{Th:active}. The following corollary is a direct consequence of Theorem~\ref{T:variationBound}, Theorem~\ref{Th:active}, Theorem~\ref{T:general}, and Proposition~\ref{P:generalXP}.

\begin{corollary}\label{C:gen}
    \CFR, \ACR, \SCR, and \textsc{Generalized CnR} are \FPT prameterized by $\mathsf{vcn}$. Specifically, each of these variants is solvable in $(\mathsf{vcn} \cdot 2^\mathsf{vcn}+\mathsf{vcn})^{\mathsf{vcn}+1}\cdot n^{\mathcal{O}(1)}$ time.
\end{corollary}

\section{Polynomial Kernels for \CR}
In this section, we provide a linear kernel for \CR parameterized by the neighbourhood diversity ($\mathsf{nd}$) of the input graph. One of the key benefits of $\mathsf{nd}$ as a parameter is that it is computable in polynomial time~\cite{ndCompute}. More specifically, in polynomial time, we can compute a minimum partition of $V(G)$ into classes $V_1,\ldots, V_w$ such that each $V_i$ contains vertices of the same type. Hence, a linear kernel parameterized by $\mathsf{nd}$ can be very useful from an applicative perspective. 

Since for any two vertices $u,v\in V_i$, for $i\in [w]$, $N(v) \setminus \{u\} = N(v)\setminus \{ u\}$, we have that either each $V_i$ is an independent set ($N(v) = N(u)$ in this case) or each $V_i$ is a clique ($N[v] = N[u]$ in this case). Now, we use the following reduction rules. 
\begin{RR}[RR\ref{nd1}]\label{nd1}
    If $k\geq w$, then answer positively.
\end{RR}

We have the following lemma to prove that RR\ref{nd1} is safe.
\begin{lemma}\label{L:ndBound}
    For a graph $G$, $\mathsf{c}(G) \leq \mathsf{nd}$.
\end{lemma}
\begin{proof}
    Let $S$ be a set of vertices such that $S$ contains exactly one vertex, say $v_i$, from each neighbourhood class $V_i$. Then, (since we assume $G$ to be connected) observe that $S$ is a dominating set of $G$. Hence, the cops have a trivial winning strategy by placing a cop on each vertex of $S$ (and $|S| \leq w$). Therefore, $\mathsf{c}(G) \leq \mathsf{nd}$.
\end{proof}

Next, if $k=1$, then we apply RR\ref{R:KV2} (the XP-algorithm for \CR from Proposition~\ref{P:XP}). Hence, we assume that $k\geq 2$. Next, we have the following reduction rule. 

\begin{RR}[RR\ref{nd2}]\label{nd2}
    For each neighbourhood class $V_i$, keep one arbitrary vertex and delete the rest.
\end{RR}

We have the following lemma to prove that RR\ref{nd2} is safe.

\begin{lemma}\label{L:nSafe}
    Let $V_i= \{v_1,\ldots,v_\ell\}$ be a neighbourhood class of $G$ having at least two vertices ($\ell \geq 2$). Consider the subgraph $H$ of $G$ induced by $V(G)\setminus \{v_\ell\}$. Then, for $k>1$, $G$ is $k$-copwin if and only if $H$ is $k$-copwin.
\end{lemma}
\begin{proof}
    We have the following two cases depending on whether $V_i$ is an independent set or a clique.
    \begin{enumerate}
        \item $V_i$ is an independent set: Note that, in this case, $N(v_\ell) = N(v_1)$. Therefore, due to Lemma~\ref{L:VCkernel}, we have that $G$ is $k$-copwin if and only if $H$ is $k$-copwin.

\item $V_i$ is a clique: Note that in this case, $N[v_\ell] = N[v_1]$. The proof, in this case, (specifically the forward direction) follows from arguments presented in the proof of Lemma~\ref{L:VCkernel}. For the reverse direction, here for $x\neq v_\ell$, $I(x) = x$ and $I(v_\ell) = v_1$. Now, note that every possible move of $\R$ in $G$ can be mapped to a valid move of the \textit{image of the robber} in $H$, just like in the proof of Lemma~\ref{L:VCkernel}. The only difference here is that when $\R$ is at $v_\ell$ (image of $\R$ is at $v_1$), $\R$ can move to $v_1$ as well (along with vertices in $N(v_1)$). Notice that this move can be translated to the movement of the image of $\R$ in $H$ where the image of $\R$ chooses to stay on the same vertex in its move. Hence, the cops will first capture the image of $\R$ in $H$, and then capture $\R$ in $G$. 
    \end{enumerate}
This completes the proof of this lemma.    
\end{proof}

Since we keep only one vertex of each type and there are at most $w$ types, we have the following theorem.

\ND*
We have the following corollary as a consequence of Theorem~\ref{T:nd}.
\begin{corollary}
    \CR is \FPT parameterized by $\mathsf{nd}$. Specifically, it is solvable in $\mathsf{nd}^\mathsf{nd}\cdot n^{\mathcal{O}(1)}$ time.
\end{corollary}

Moreover, it is not difficult to see that this kernelization can be extended to \LCR and \CAR using an extension of Lemma~\ref{L:VCkernelVariants}, giving us a kernel with at most $\mathsf{nd}$ vertices. Moreover, using a reduction rule similar to RR\ref{R:G} where we keep $k$ vertices of each type, we can have a kernel with at most $k\cdot \mathsf{nd}$ vertices for \CFR and \ACR.
We have the following lemma, for which we provide a proof outline.

\begin{lemma}\label{L:ND2}
    Let $V_i = \{v_1,\ldots,v_\ell\}$ is a neighbourhood class of $G$ containing at least $k+2$ vertices $(\ell \geq k+2)$. Consider the subgraph of $G$ induced by $V(G)\setminus \{v_\ell\}$. Then, for $k>1$, $s\geq 1$, and $x\in \{active, s\}$, $\mathsf{c}_x(G) \leq k$  if and only if $\mathsf{c}_x(H) \leq k$.
\end{lemma}
\begin{proofsketch}
    Similar to the proof of Lemma~\ref{L:nSafe}, we have the following two cases depending on whether $V_i$ is an independent set or a clique.
    \begin{enumerate}
        \item $V_i$ is an independent set: Proof of this case follows from the proof of Lemma~\ref{L:Gen}.

        \item $V_i$ is a clique: Here, for each $v_j \in V_i$, $N[v_j] = N[v_\ell]$. For each vertex $x \neq v_\ell$, let $I(x) = x$ and $I(v_\ell) = v_1$.  
        
        First, let $\mathsf{c}_x(G) \leq k$. Then, we use the strategy of the cops from $G$ to capture $\R$ in $H$ with the only change that whenever a cop, say, $\C_i$, wants to move to a vertex $x$ in $G$, it moves to  $I(x)$ in $H$ instead with the only contingency that if $\C_i$ wants to move from $v_1$ to $v_\ell$, then it moves to $v_2$ (so that if the cops are active, then this is indeed a valid move of $\C_i$ in $H$). Observe that the cops can capture $\R$ in $G$ using this strategy even when the cops are restricted to the vertices of $H$. Hence, the cops can capture $\R$ using this strategy in $H$. 

        In the reverse direction, let $\mathsf{c}_x(H) \leq k$. Note that if $k$ active cops have a winning strategy against a flexible robber in $G$, then $k$ active cops have a winning strategy against an active robber in $G$ as well. Hence, for the ease of arguments, we show that $k$ active cops have a winning strategy in $G$ even if $\R$ is flexible to show that $\mathsf{c}_{active} \leq k$. The cops assume that the \textit{image of} $\R$ is occupying the vertex $I(x)$ when $\R$ is occupying the vertex $x$.  Thus, we have an image of $\R$ moving in $H$ with the same capabilities as $\R$. The cops will capture the image of $\R$ using their winning strategy from $H$. Notice that once the image of $\R$ is captured, if $\R$ is at a vertex $x \neq v_\ell$, then $\R$ is captured as well. Otherwise, $\R$ is at $v_\ell$ and there is some cop, say, $\C_1$, at  $v_1$. In the case of \ACR, $\R$ will be captured in the next move of cops (since $v_1v_\ell \in E(G)$). In the case of \CFR, if this is a cop move (that is, the image of $\R$ is captured on a robber move), then $\C_1$ will capture $\R$ in the next move. Otherwise, in the previous move of the cops, $\C_1$ moved to $v_1$ while $\R$ was at $v_\ell$. In this case, since $N[v_1] = N[v_\ell]$, $\C_1$ could have moved to $v_\ell$ to capture $\R$ itself.  Hence, $\mathsf{c}_x(G) \leq k$.   
    \end{enumerate}
    This completes the proof.
\end{proofsketch}

Since $\mathsf{c}(G) \leq \mathsf{nd}$ for all of these variants (as there is a dominating set of size $\mathsf{nd}$ in $G$), we have the following result as a consequence of Lemma~\ref{L:ND2} (and arguments presented above).

\begin{theorem}
    \LCR and \CAR parameterized by $\mathsf{nd}$ admit a kernel with at most $\mathsf{nd}$ vertices. Moreover, \CFR and \ACR parameterized by $\mathsf{nd}$ admit a kernel with at most $\mathsf{nd}^{2}$ vertices.
\end{theorem}

Finally, we remark that this technique of kernelization will not work directly for \SCR. For example, consider a complete graph on $n$ vertices, for which $\mathsf{nd} = 1$ (all the vertices have the same type) and $\mathsf{c}_{surround} = n$, and if we remove any vertex from this clique, $\mathsf{c}_{surround}$ decreases. Moreover, as evident from our example of complete graphs, $\mathsf{c}_{surround}$ cannot be bounded by any computable function that depends only on $\mathsf{nd}$.

\section{Incompressibility}\label{S:NPoly}
\subsection{Incompressibility of \CR}
In this section, we show that it is unlikely that the \CR  problem parameterized by $\mathsf{vcn}$ admits a polynomial compression. For this purpose, we first define the following problem. In \RBDS, we are given a bipartite graph $G$ with a vertex bipartition $V(G) = T \cup N$ and a non-negative integer $k$. A set of vertices $N'\subseteq N$ is said to be an \textit{RBDS} if each vertex in $T$ has a neighbour in $N'$. The aim of \RBDS is to decide whether there exists an \textit{RBDS} of size at most $k$ in $G$. Accordingly, we have the following decision version of \RBDS. 


\Pb{\RBDS}{A bipartite graph $G$ with vertex bipartition $V(G) = T \cup N$, and a non-negative integer $k$.}{Question}{ Does $G$ has an \textit{RBDS} of size at most $k$?}

Dom, Lokshtanov,  and Saurabh~\cite{RBDSIncompressibility} proved that it is unlikely for \RBDS parameterized by $|T|+k$ to admit a polynomial compression. More precisely, they proved the following result.



\begin{proposition}[\cite{RBDSIncompressibility}]\label{R:RBDS}
\RBDS parameterized by $|T|+k$ does not admit a polynomial compression, unless \NPoly.
\end{proposition}
 
We show that \CR parameterized by the $\mathsf{vcn}$ does not have a polynomial compression by developing a polynomial parameter transformation from  \RBDS parameterized by $|T|+k$ to \CR parameterized by  $\mathsf{vcn}$. 

\subsubsection{Bipartite Graphs with Large Degree and Girth}
For our reduction, we borrow a construction by Fomin at al.~\cite{fomin} of bipartite graphs having high girth and high minimum degree, which they used to prove NP-hardness (and $W[2]$-hardness for the solution size $k$) of \CR. 
For positive integers $p,q$, and $r$, we can construct a bipartite graph $H(p,q,r)$ with $rqp^2$ edges and a bipartition $(X,Y)$, with $|X| = |Y| = pq$. The set $X$ is partitioned into sets $U_1, \ldots, U_p$, and the set $Y$ is partitioned into sets $W_1, \ldots W_p$, with $|U_i| = |W_i| = q$. By $H_{i,j}$ we denote the subgraph of $H(p,q,r)$ induced by $U_i \cup W_j$, and by $\mathsf{deg}_{i,j}(z)$ we denote the degree of vertex $z$ in $H_{i,j}$. Fomin et al.~\cite{fomin} provided the following construction:

\begin{proposition}[\cite{fomin}]
Let $q \geq 2p(r+1) \frac{(p(r+1)-1)^6-1}{(p(r+1)-1)^2-1}$. Then, we can construct $H(p,q,r)$ in time $\mathcal{O}(r\cdot q \cdot p^2)$ with the following properties.
\begin{enumerate}
    \item The girth of $H(p,q,r)$ is at least 6.
    \item For every vertex $z \in V(H_{i,j})$ and every $i,j \in [p]$, we have $r-1 \leq \mathsf{deg}_{i,j}(z) \leq r+1$.
\end{enumerate}
\end{proposition}

\subsubsection{Polynomial Parameter Transformation}

Suppose that we are given an instance $(G,k)$ with $V(G) = T \cup N$ of the \RBDS problem. First, we construct a graph $G'$ with $V(G') = T' \cup N'$ from $G$ by introducing two new vertices, $x$ and $y$, such that $T' = T \cup \{x\}$ and $N' = N \cup \{y\}$, and $E(G') = E(G) \cup \{xy \}$. We have the following observation.

\begin{observation}\label{O:reduction}
$G$ has an \textit{RBDS} of size at most $k$ if and only if $G'$ has an \textit{RBDS} of size at most $k+1$. Moreover, any \textit{RBDS} of $G'$ contains $y$.
\end{observation}

Now, we present the main construction for our reduction. Denote the vertex set  $V(T')$ by $\{v_1, v_2, \ldots, v_{p'}, x\}$. Moreover, let $p = p'+1$, $\ell=k+1$, $r = \ell+2$, and $q = \lceil 2p(r+1) \frac{(p(r+1)-1)^6-1}{(p(r+1)-1)^2-1} \rceil$. 

We construct $H(p,q,r)$ such that each of $U_i$ and $W_i$, for $0 < i \leq p'$, contains $q$ copies of vertex $v_i$, and each of $U_p$ and $W_p$ contains $q$ copies of vertex $x$. Now, we obtain a graph $G''$ by adding one more set of vertices  $P$ to $H(p,q,r)$ such that $V(P) = V(N')$. Moreover, if there is an edge between a vertex $u \in N'$ and a vertex $v_i \in T'$, then we add an edge between $u$ and every vertex of $U_i$, and also between $u$ and every vertex of $W_i$. Similarly, we add an edge between $y$ and every vertex of $U_p$, and between $y$ and every vertex of $W_p$. Finally, we make the vertex $y$ adjacent to every vertex of $P$. See Figure~\ref{fig:NoPoly} for reference.

For correctness, we have the following lemma.

\begin{figure}
    \centering
    \includegraphics[scale=0.65]{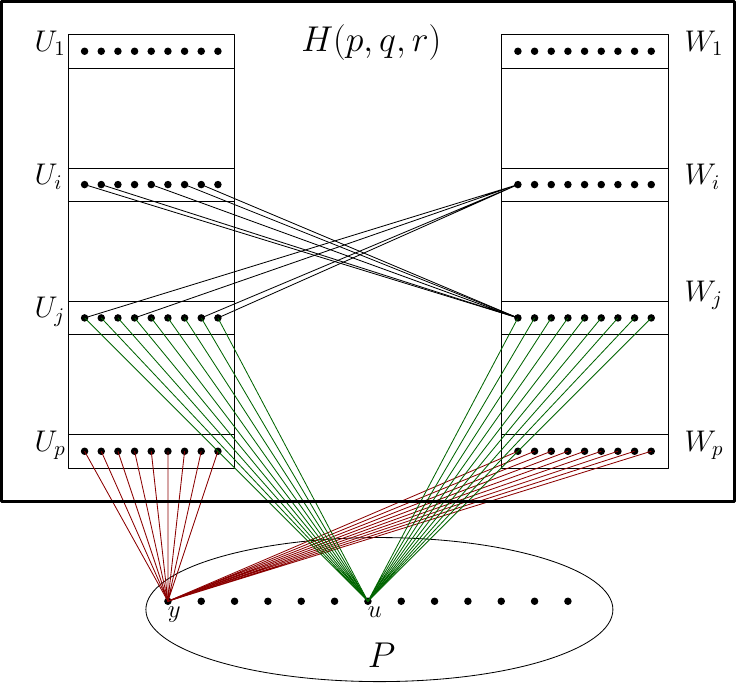}
    \caption{Illustration for $H(p,q,r)$ and $P$. If a vertex $u$ is connected to $v_j$ in $G$, then here $u$ is connected to every vertex of $W_j$ and $U_j$. Moreover, for every $i,j$, each vertex in $U_i$ has at least $r-1$ neighbors in $U_j$.}
    \label{fig:NoPoly}
\end{figure}

\begin{lemma}\label{L:NoPoly}
 $G'$ has an \textit{RBDS} of size at most $\ell$ if and only if  $G''$ is $\ell$-copwin.  
\end{lemma}
\begin{proof}
First, we show that if $G'$ has an \textit{RBDS} of size $\ell$, then $\ell$ cops have a winning strategy in $G''$. Let $S\subseteq N'$ be an RBDS in $G'$ of size at most $\ell$. The cops begin by choosing the vertices corresponding to $S$ in $P$. Observe that the vertex $y$ has to be present in $S$. Since vertex $y$ dominates each vertex in $P$, the robber cannot safely enter a vertex in $P$. Additionally, due to the construction of $G''$, the vertices of $S$ dominate each vertex in $H$. Hence, the robber cannot safely enter a vertex in $H$. Therefore, the robber will be captured in the first move of the cops.

Next, we show that if there is no \textit{RBDS} of size $\ell$ in $G'$, then $\ell$ cops do not have a winning strategy. We prove this by giving a winning strategy for the robber. First, we show that the robber can safely enter the graph. In the beginning, let there be $\ell_1 \leq \ell$ cops in $P$ and $\ell_2 \leq \ell$ cops in $H$. Since there is no \textit{RBDS} of size $\ell$ in $G'$, for every placement of at most $\ell$ cops in $P$, there exists at least one pair of $U_i$ and $W_i$ such that no vertex of $U_i$ and $W_i$ is dominated by the cops from $P$. Let $U_i$ and $W_i$ be one such pair of sets such that no vertex of $U_i$ and $W_i$ is dominated by the cops from $P$. Moreover, since each vertex of $H$ can dominate at most $p(r+1)$ vertices in $H$, $\ell_2$ cops can dominate at most $\ell \cdot p(r+1)$ vertices. Since $U_i$ (and $W_i$ also) contains $q$ vertices, and $q> \ell \cdot p(r+1)$, the $\ell_2$ cops in $H$ cannot dominate all vertices of $U_i$, and hence the robber can safely enter a vertex of $U_i$. 

Now, whenever the robber is under attack, it does the following. Without loss of generality, let us assume that the robber is in $U_i$ (the case of $W_i$ is symmetric). Since there are at most $\ell$ cops in $P$, there is always a $W_j$ such that no vertex of $W_j$ is dominated by cops from $P$. Since each vertex in $U_i$ has at least $r-1 = \ell+1$ neighbours in $W_j$,  the robber can move to at least $\ell+1$ vertices of $W_j$. Since the girth of $H$ is at least 6, no vertex from $H$ can dominate two vertices of $W_j$ that are adjacent to the robber; else, we get a cycle on four vertices. Hence, at most $\ell$ cops from $H$ can dominate at most $\ell$ neighbors of the robber in $W_j$, and the robber has at least $\ell+1$ neighbors in $W_j$.  Hence, the robber can move to a safe vertex in $W_j$. Since the graph $H$ is symmetric, the robber can move safely from  $W_{j'}$ to $W_{i'}$ also. The robber follows this strategy to avoid capture forever.

This completes the proof of our lemma.
\end{proof}

Next, we have the following observation to show that there exists a vertex cover $U$ of $G''$ such that $|U| = poly(|T|,k)$. 

\begin{observation}\label{O:SizeBoundReduction}
$V(H) \cup \{y\}$ is a vertex cover of $G''$. Therefore, the vertex cover number of $G''$ is at most $2\cdot p \cdot q+1 = 1+ 2p\cdot \lceil 2p(k+3) \frac{(p(k+3)-1)^6-1}{(p(k+3)-1)^2-1}\rceil$, where $p = |T|+1$.
\end{observation}

This completes the proof of the argument that \CR parameterized by the $\mathsf{vcn}$ is unlikely to admit a polynomial compression. Thus, we have the following theorem as a consequence of Lemma~\ref{L:NoPoly}, Observation~\ref{O:SizeBoundReduction} and Proposition~\ref{R:RBDS}.

\VCNPoly*

We prove the incompressibility of the variants (Theorem~\ref{th:variantPolyCOmpression}) in the Appendix.

\subsection{Incompressibility for Variants}
In this section, we prove Theorem~\ref{th:variantPolyCOmpression}. In Theorem~\ref{th:PolyCompression}, we proved  that it is unlikely for \CR to admit a polynomial compression. For this purpose, we constructed a graph $G''$ where $k$ cops have a winning strategy if and only if the graph $G'$ has an $RBDS$ of size at most $k$. If $G'$ has an $RBDS$ of size $k$, then there is a dominating set of size $k$ in $G''$. Else, there is no winning strategy for $k$ cops in $G''$. Here, we  use the same construction to show that the variants we study (except for \SCR) are unlikely to admit a polynomial compression parameterized by $\mathsf{vcn}$. We establish this by proving that $G''$ is $k$-copwin for these variants if and only if $G'$ has an $RBDS$ of size at most $k$.

As discussed earlier, for a graph $G$, $\mathsf{c}(G) \leq \mathsf{c}_{lazy}(G)$,  $\mathsf{c}(G) \leq \mathsf{c}_{attacking}(G)$, and $\mathsf{c}(G) \leq \mathsf{c}_{s}(G)$ (for any $s\geq 1$). Therefore, if $G$ does not have an $RBDS$ of size at most $k$, then $\mathsf{c}(G)>k$, and hence, $\mathsf{c}_{lazy}(G) >k$, $\mathsf{c}_{attacking}(G) > k$, and $\mathsf{c}_{s}(G) > k$ (for $k>0$). To see the reverse direction, observe that in each of these three variants, if the cops start by occupying a dominating set, then they win in the next round. Hence, this establishes that it is unlikely for \CAR, \CFR, \LCR parameterized by $\mathsf{vcn}$ to admit to admit a polynomial compression. 

Similarly, it is true for \ACR also, that if the cops start by occupying a dominating set, then they win in the next round. Hence, we have to only show that if there is no $RBDS$ of size $k$ in $G'$ (and hence, no dominating set of size $k$ in $G''$), then $k$ cops do not have a winning strategy in $G$ for \ACR.  The robber uses the following strategy. When $\R$ is under attack, it follows the strategy from \CNR.  Now, $\R$ is forced to move (because it is active) even when it is on a safe vertex. Note that $\R$ always stays in $H(p,q,r)$. Due to symmetry, let us assume it is in some vertex $v$ in some block $U_i$. In this case, $\R$ can simply move to a vertex in $W_i$. Observe here that since vertices in $U_i$ are safe, the vertices in $W_i$ are also safe.   

Thus, we have the following lemma to establish that these variants are unlikely to admit a polynomial compression.

\begin{figure}
    \centering
    \includegraphics[scale=0.6]{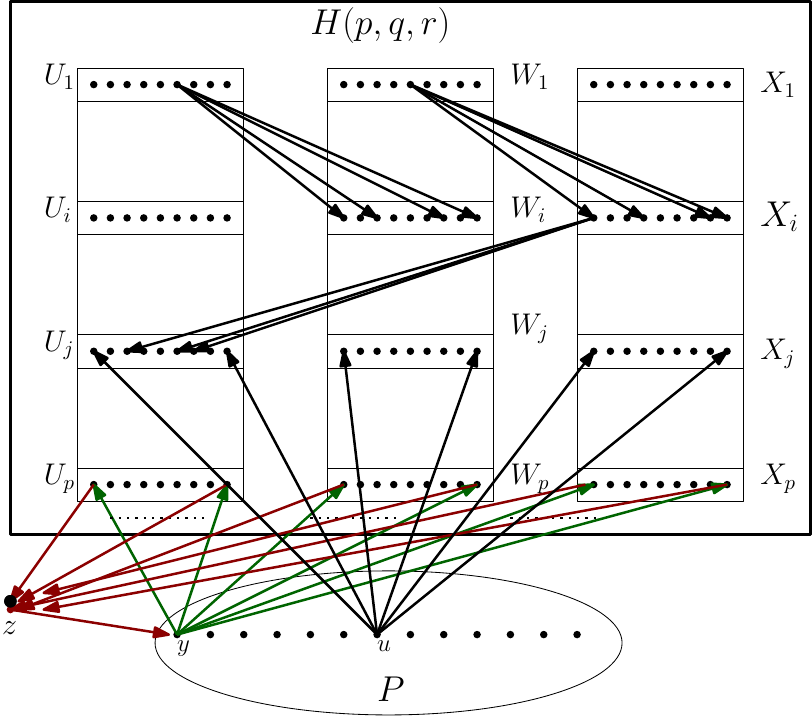}
    \caption{Illustration for $H(p,q,r)$ and $P$. If a vertex $u$ is connected to $v_j$ in $G$, then for every vertex $v \in W_j\cup X_j \cup U_j$, we orient the arc $\overrightarrow{uv}$. For each vertex $v \in V(P)\setminus \{y\}$, we orient the arc $\overrightarrow{yv}$. We add a vertex $z$, and arcs $\overrightarrow{zy}$ and $\overrightarrow{xz}$ for each $x\in U_p \cup X_p \cup W_P$.}
    \label{fig:appendix}
\end{figure}


\begin{lemma}\label{L:variantPolyCOmpression}
\LCR, \CAR, \CFR, and \ACR parameterized by the vertex cover number do not admit a polynomial compression, unless \NPoly.
\end{lemma}

This result can also be extended to directed (or oriented) graphs. We have the following lemma.
\begin{lemma}\label{L:directedPolyCompression}
\CR on strongly connected directed and oriented graphs parameterized by vertex cover number does not admit a polynomial compression, unless \NPoly.
\end{lemma}
\begin{proof}
For the case of directed graphs, we can simply replace each edge in the construction with a loop edge (directed cycle on two vertices). 

To prove this result for oriented graphs, we do the following. Here, we change the underlying graph $G''$. First, instead of having two partitions $U$ and $W$, we have three partitions $U,W$, and $X$ (with the same rules). See Figure~\ref{fig:appendix} for an illustration. Second, we add edges between $U$ and $W$, $W$ and $X$, and $X$ and $U$ following the rules of the construction. Moreover, the edge rules for vertices in $P$ are the same (that is, if a vertex has edges with each vertex in some $U_i$, it has edges with each vertex in $W_i$ and $X_i$ as well). Next, we define orientations. For the vertex $y$, we orient all the edges as outgoing. For every vertex $u\in P \setminus \{y\}$, we mark all the edges as outgoing, except for the edge $uy$ (which is oriented $\overrightarrow{yu}$). For each edge $uv$ such that $u\in U$ and $w \in W$, orient the edge $\overrightarrow{uw}$. For every edge $wx$ such that $w\in W$ and $x\in X$, orient the edge $\overrightarrow{wx}$. For each edge $xu$ such that $x\in X$ and $u\in U$, orient the edge $\overrightarrow{xu}$. Finally, add an extra vertex $z$, and add arc $\overrightarrow{zy}$. Moreover, for each vertex $v \in U_p \cup W_p \cup X_p$, add an arc $\overrightarrow{vz}$. 

It is straightforward to see that $\overrightarrow{G''}$ is a strongly connected oriented graph. Moreover, if $G''$ has a dominating set of size $k$, then $k$ cops have a winning strategy by occupying these vertices in $\overrightarrow{G''}$. Observe that, at this point, $\R$ can enter only at vertex $z$ and cannot move as long as there is a cop, say, $\C_1$, at $y$ (which there is due to the construction of $G'$ and $G''$). Now, since $\overrightarrow{G''}$ is strongly connected, some other cop, say, $\C_2$, can move to capture $\R$ in a finite number of rounds. For the reverse direction, if $G'$ does not have an $RBDS$ of size $k$ (and hence $G''$ does not have a dominating set of size $k$), then following the arguments of Lemma~\ref{L:NoPoly}, $\R$ can enter at a safe vertex in $U$. Then, whenever $\R$ is under attack, it can move to a safe vertex in $W$. Similarly, it can move from $W$ to $X$ and from $X$ to $U$ when under attack.  Moreover, note that vertex $z$ does not attack any vertex in $U\cup W\cup X$. Hence, $\R$ has a clear evading strategy.

This completes our proof.
\end{proof}

Lemma~\ref{L:variantPolyCOmpression} and Lemma~\ref{L:directedPolyCompression} directly imply Theorem~\ref{th:variantPolyCOmpression}.

\section{Conclusion and Future Directions}\label{S:conclude}

In this paper, we conducted a comprehensive analysis of the parameterized complexity of \CR parameterized by $\mathsf{vcn}$. 
First, we showed that the cop number of a graph is upper bounded by $\frac{\mathsf{vcn}}{3}+1$. Second, we proved that \CR parameterized by $\mathsf{vcn}$ is \FPT by designing an exponential kernel. We complemented this result by proving that it is unlikely for \CR parameterized by $\mathsf{vcn}$ to admit a polynomial compression. We then extended these results to other variants as well as to other parameters.

To achieve our kernelization results, the rules we used concerned removing (false or true) twins from the graph. These rules are easy to implement and hence can be used to reduce the complexity of the input graph, even when the input graph is far from the considered parameters. For example, for cographs, none of the considered parameters is constant/bounded, but cographs can be reduced to a single vertex with the operation of removing twins, and hence, our reduction rules give an alternate proof that the cop number of cographs  is at most two~\cite{joret} for several variants. Moreover, \textsc{MTP} is well-studied with the motivation of designing computer games. Some examples of these variants include: multiple targets and multiple pursuer search~\cite{MTP1} with applications in controlling non-player characters in video games; \textsf{MTP} from the robber's perspective with faster cops~\cite{MTP3} where the strategies were tested on Baldur's Gate; \textsf{MTP} modeled with edge weights and different speeds of agents~\cite{app1} with the specific examples of Company of Heroes and Supreme Commander. Moreover, the PACMAN game's movement can be considered as an instance of \ACR on a partial grid. One of the key aspects of designing these games is to come up with scenarios that are solvable but look complex and challenging.  Our reduction rule can help in this regard. One can begin with an easy-to-resolve instance of \CR, and then keep adding  twins to this instance (recursively) to get an instance that looks sufficiently complex but has the same complexity.

Finally, we defined a new variant of \CR, named \textsc{Generalized CnR}, that generalizes many well-studied variants of \CR including \CFR, \ACR, \textsc{Cops and Robber From a Distance}~\cite{bonatoDistance}, and also generalizes the games of~\cite{MTP3,app1}. We showed that RR\ref{R:G} provides a kernel for \textsc{Generalized CnR} as well. This gives hope that RR\ref{R:G} can be used to get kernels for many practical variants not explicitly studied in this paper.

Still, many questions on the parameterized complexity of \CR remain open. We list some of these questions below.
\begin{question}

Does there exist an \FPT algorithm for \CR parameterized by $\mathsf{vcn}$ with running time $2^{\mathcal{O}(\mathsf{vcn})} \cdot n^{\mathcal{O}(1)}$?
\end{question}

\begin{question}
Does there exist a better bound for the cop number with respect to $\mathsf{vcn}$? In particular, is $\mathsf{c}(G) = o(\mathsf{vcn})$?

\end{question}

\begin{question}
Does \CR parameterized by $\mathsf{vcn}$ admit a polynomial $\alpha$-approximate kernel? 
\end{question}


\begin{question}
Study \CR with respect to the following parameters: (1) feedback vertex set (2) treewidth (3) treedepth. In particular, is \CR \FPT parameterized by treewidth?
\end{question}

\bibliographystyle{plainurl}

\bibliography{main}
\end{document}